\documentclass[10pt, conference, letterpaper]{IEEEtran}

\usepackage{booktabs}
\usepackage{amsmath,amssymb,bm,dsfont,bbm}
\usepackage{algorithmic}
\usepackage{amsthm,nccmath,lipsum,cuted}
\usepackage[dvips]{color}
\usepackage{graphicx}
\usepackage[lined,boxed,commentsnumbered, ruled]{algorithm2e}
\usepackage{subcaption}
\usepackage{amsfonts}
\usepackage{ragged2e}
\newtheorem{theorem}{Theorem}
\newtheorem{lemma}{Lemma}
\newtheorem{definition}{Definition}

\newtheorem{corollary}{Corollary}

\usepackage{graphicx}
\usepackage{verbatim}
\usepackage{xcolor}
\usepackage[footnote,draft,danish,silent,nomargin]{fixme}
\usepackage[colorlinks=true,linkcolor=black]{hyperref}
\def\P{{\mathbb P}}  %%%   appears in many equations  Prob
\def\E{{\mathbb E}}  %%%   for "expectation value" (in math mode)

\def\ceild{{\frac{\lceil \lambda d \rceil}{\lambda}}}
\def\TD{T_\text{D}}
\def\TA{T_\text{A}}

\graphicspath{{./figs/}}

\IEEEoverridecommandlockouts

\begin{document}

\title{On the Distribution of AoI for the GI/GI/1/1 and GI/GI/1/2* Systems: Exact Expressions and Bounds}
 \author{
    \IEEEauthorblockN{Jaya Prakash Champati\IEEEauthorrefmark{1}, 			     Hussein Al-Zubaidy\IEEEauthorrefmark{2}, 
    James Gross\IEEEauthorrefmark{1}}
    \IEEEauthorblockA{\IEEEauthorrefmark{1} Information Science and Engineering, KTH Royal Institute of Technology, Stockholm, Sweden}
    \IEEEauthorblockA{\IEEEauthorrefmark{2} Guest researcher, Network and Systems Engineering, KTH Royal Institute of Technology, Stockholm, Sweden}
$\{$jpra,hzubaidy$\}$@kth.se, james.gross@ee.kth.se   
\thanks{This work appeared in the proceedings of IEEE INFOCOM, 2019~\cite{champati:Infocom2019}. The current paper complements the INFOCOM version by including the missing proofs and some additional explanations.}
\thanks{This work has been partially supported by the Swedish Research Council VR under grant 2016-04404.}
}
\maketitle

\begin{abstract}
Since Age of Information (AoI) has been proposed as a metric that quantifies the freshness of information updates in a communication system, there has been a constant effort in understanding and optimizing different statistics of the AoI process for classical queueing systems. In addition to classical queuing systems, more recently, systems with no queue or a unit capacity queue storing the latest packet have been gaining importance as storing and transmitting older packets do not reduce AoI at the receiver. Following this line of research, we study the distribution of AoI for the GI/GI/1/1 and GI/GI/1/2* systems, under non-preemptive scheduling. For any single-source-single-server queueing system, we derive, using sample path analysis, a fundamental result that characterizes the AoI violation probability, and use it to obtain closed-form expressions for D/GI/1/1, M/GI/1/1 as well as systems that use zero-wait policy. Further, when exact results are not tractable, we present a simple methodology for obtaining upper bounds for the violation probability for both GI/GI/1/1 and GI/GI/1/2* systems. An interesting feature of the proposed upper bounds is that, if the departure rate is given, they overestimate the violation probability by at most a value that decreases with the arrival rate. Thus, given the departure rate and for a fixed average service, the  bounds are tighter at higher utilization. 
\end{abstract}

\section{Introduction}\label{sec:intro}
In the recent past, there is an ever-increasing demand for networked systems that support emerging time-critical control applications. These applications include,  among many others, smart grid, factory automation and augmented reality. A basic building block in these applications is a closed-loop control which in its simplest form comprises: 1) a source (e.g. sensor) that samples a process of interest and transmits the status updates or packets, 2) a receiver (e.g. controller/monitor), and 3) an actuator. In such a control-loop a status update received after certain duration of its generation may become stale as the control decision based on this sample may lead to untimely action by the actuator. Thus, the \textit{freshness} of the status updates at the receiver plays a key role in the design of such networked systems. Age of Information (AoI) has been proposed as a relevant metric to quantify the freshness of the information~\cite{kaul_2011a}. It is defined as the time elapsed since the generation of the latest status update that is received at the receiver. When a status update is received, AoI gives the time elapsed since its generation thus indicating the freshness of the status update.

%Motivation for GI/GI/1/1 and GI/GI/1/2*
%Since guarantees on the freshness of the status updates play a crucial role in the deisgn of the networked systems, 
%Providing QoS guarantees for AoI has received. 
In contrast to latency, AoI has interesting property that it increases at both low and high sampling rate for queueing systems using First-Come-First-Serve (FCFS) policy~\cite{kaul_2012b}. This property led to initial works focusing on quantifying and minimizing the average AoI for the M/M/1, M/D/1 and D/M/1 queues in~\cite{kaul_2012b}, and for an M/M/1 queue with multi-sources in~\cite{yates_2012a}, under FCFS policy. Subsequent works~\cite{kaul_2012a,Najm_2016} have studied Last-Come-First-Serve (LCFS) policy as it reduces AoI compared with FCFS policy. In fact, the authors in~\cite{Bedewy_2017a} proved that in the domain of non-preemptive scheduling policies LCFS minimizes the AoI process, in stochastic ordering sense. One may further reduce AoI compared to  LCFS policy by considering packet discarding. Intuitively, even if an infinite capacity queue is available, when the server is busy, discarding the arriving packets except for the most recent packet would result in a lower AoI when compared with storing and transmitting any older packets. This motivated research efforts toward studying systems with no queue or a single capacity queue storing the latest packet~\cite{Costa_2016,Najm_2018,Soysal_2018,Yoshiaki2018}. Following this line of research we study the GI/GI/1/1 and GI/GI/1/2* systems. 

The GI/GI/1/1 system has no queue and an arrival is discarded if the server is busy, otherwise it is served immediately. The GI/GI/1/2* system has a queue with unit capacity. Different from GI/GI/1/2, whenever a packet arrives, it \textit{replaces} the packet that is in the queue, or will be served if the server is idle. Even though these systems look rudimentary, they are very important models for systems that are driven by the AoI metric. Note that AoI is only reduced upon a packet departure that has a generation time later than that of the previously departed packet. From here on we refer to a packet that reduces AoI upon its departure as \textit{information update packet}. In both the above systems every packet served is an information update packet as they always serve the most recently generated packet. Furthermore, while GI/GI/1/1 naturally arises in systems with no queue, GI/GI/1/2* is an attractive choice among work-conserving single-server queueing systems with non-zero queue capacity as it only stores most recently generated packet in the queue. Finally, the statistics obtained for GI/GI/1/1 and GI/GI/1/2* systems can be immediately used for obtaining the statistics for a system using zero-wait policy/just-in-time updates~\cite{Costa_2016}.
%, where a status update is only generated when the service of previous update is finished.
%Intuitively, storing and serving packets that have older generation time will only increase AoI. 

%Following their work, we study two packet management policies for general inter-arrival time and service-time distributions and : 1) Discard In a GI/GI/1/1 system there is no queue. An arriving packet is dropped if the server is busy, else will get served immediately. In a GI/GI/1/2* system the capacity of queue is one and an arriving packet always replaces the packet in the queue is any or is served if the server is idle. In particular, they consider three packet policies, 1) transmissions with no queue M/M/1/1 and a unit capacity queue with more recent packet replacing the exisi. 

%Motivation for studying the distribution of AoI - do we need this. May be not as special paragraph but a simple sentence
Despite their significance, existing results for the GI/GI/1/1 and GI/GI/1/2* systems are limited. In~\cite{Costa_2016}, the authors studied the M/M/1/1 and M/M/1/2* systems, and computed the average AoI and the distribution of the peak AoI. Closed-form expressions for the average AoI and the average peak AoI are derived for an M/GI/1/1 system for a single source in~\cite{Najm_2017}, and for multiple sources in~\cite{Najm_2018}, under preemptive scheduling. Noting that exact expressions for non-exponential interarrival times are difficult to obtain, the authors in~\cite{Soysal_2018} derived upper bounds for the average AoI for a GI/GI/1/1 system. To the best of our knowledge, neither exact expressions nor bounds exist for the distribution of AoI for a GI/GI/1/1 system or any special case of it. On the other hand, for a GI/GI/1/2* system, the authors in~\cite{Yoshiaki2018} derived exact expressions for the characteristic function of AoI for special cases where either the inter-arrival times or the service times are exponential, i.e., for the GI/M/1/2* and M/GI/1/2* systems. The authors achieve this by proposing a general formula for the distribution of AoI for any single-server-single-source queueing system.  
%The authors achieve this by proposing a general formula that relates the distribution of AoI with the distribution of peak AoI and the distribution of system delay, under a mild assumption of stationarity and ergodicity of the AoI process. 
Nevertheless, if neither the service times nor the inter-arrival times are exponential, the general formula becomes intractable and in this case no known results exist for the GI/GI/1/2* system.
%the distribution of peak AoI and the distribution of system delay become intractable and no results exists for the GI/GI/1/2* system.
%This motivates us in studying upper bounds for the violation probability for the GI/GI/1/2* system.
%The significance of their general formula is that exact expressions for the distribution of AoI can be computed for any single source single-server queueing systems can be obtained if  

In this paper, we study the distribution of AoI for the GI/GI/1/1 and GI/GI/1/2* systems under non-preemptive scheduling. In particular, we characterize the \textit{violation probability}, i.e., the probability that AoI exceeds a given \textit{age limit} $d$. 
%For the general case where it is hard to compute the violation probability we derive upper bounds. These upper bounds on the violation probability are useful, 
This metric represents, for instance, a stochastic guarantee on the timeliness of the state information regarding a process being sampled. Assuming AoI is stationary and ergodic, we derive a fundamental result that characterizes the distribution of AoI in terms of the peak AoI process and the inter-departure time between information update packets. 
%Our result is in spirit close to the general formula proposed in~\cite{Yoshiaki2018}. 
%Nevertheless, in contrast to their approach, the characterization we provide is based on the quantity $g(k)$ which is defined as the time duration for which the AoI process is greater than the age limit $d$ in the interval between $(k-1)$th AoI peak and $k$th AoI peak. By obtaining relations between $g(k)$, waiting time and idle time of information update packets, 
As we explain in Section~\ref{sec:generalresults}, the characterization we provide is different from
the general formula proposed in~\cite{Yoshiaki2018}. Further, using this characterization, we propose and analyse upper bounds for the violation probability when the distribution of AoI is intractable.
%a unique feature of the characterization we provide is that, it not only provides exact expressions (when feasible) for AoI distribution, but also provides a simple methodology for deriving and analysing upper bounds for the violation probability when the distribution of AoI is intractable. The proposed upper bounds can be used as soft statistical guarantees for the freshness of the status updates.
Our main contributions are summarized below:
\begin{itemize}
%\item For any single-source-single-server queueing system, assuming that AoI is stationary and ergodic, we derive a fundamental relation between the distribution of AoI, peak AoI process, and the inter-departure time between information update packets. Using this characterization we propose a methodology for deriving upper and lower bounds for the violation probability of AoI. Further, these bounds only require stationarity of the AoI process.
\item For any single-source-single-server queueing system, we present a general characterization for the distribution of AoI. Using this characterization we propose a methodology to derive bounds for the violation probability. 
%Further, these bounds only require stationarity of the AoI process.
\item For the case of D/GI/1/1 and M/GI/1/1 systems we provide exact expressions for the distribution of AoI. We compute the distribution for D/M/1/1 and M/M/1/1, and show that the resulting \textit{expected} AoI expressions conform with the existing results for these special cases reported in~\cite{kaul_2012a} and~\cite{Costa_2016}. As a by-product result, we also obtain exact expressions for the case of zero-wait policy. To the best of our knowledge, these are the first known results for the distribution of AoI for the respective systems. 
\item For general inter-arrival and service time distributions, we provide upper bounds for the violation probability for both GI/GI/1/1 and GI/GI/1/2* systems. The proposed upper bounds can be used as stochastic guarantees for the freshness of the status updates in these systems.
\item We analyse the worst-case performance of the proposed upper bounds and show that, given the departure rate, the upper bounds overestimate the violation probability by at most a value that decreases as the arrival rate increases. Furthermore, the upper bounds are asymptotically tight. We emphasize that our approach for deriving and analysing upper bounds is quite general and is not restricted to GI/GI/1/1 and GI/GI/1/2* systems.
\end{itemize}

The rest of the paper is organized as follows. In Section~\ref{sec:notation} we introduce the notations used in this paper. In Section~\ref{sec:generalresults} we present general results regarding the characterization of the distribution of AoI and the bounds. In Sections~\ref{sec:GG11} and~\ref{sec:GG12} we present the results for GI/GI/1/1 and  GI/GI/1/2* systems, respectively. In Section~\ref{sec:numerical} we present some numerical results and finally conclude in Section~\ref{sec:conclusion}.

\begin{comment}
AoI has an interesting property that it increases at low and high status update sampling rate~\cite{kaul_2012b}. This has led several efforts in minimizing AoI by tuning the sampling rate. Toward this goal several works quantified expected AoI and expected peak AoI, mainly for exponential arrivals under different settings~\cite{kaul_2012b,yates_2015a,Bacinoglu_2015a}. In contrast to these works, more recently, the authors in~\cite{Champati_2018a} studied minimization of CCDF of AoI for a given age limit $d$. While all the above works focused on 

In this work we study the distribution of the Age of Information (AoI) for the GI/GI/1/1 and GI/GI/1/2* systems under non-preemptive scheduling.
\end{comment}

%\input{related}

\section{Notation and Definitions}\label{sec:notation}
%In this section we first introduce the notation for a system that consists of a single source generating status updates, and a single server work-conserving queueing system with a finite queue or infinite queue. We then describe GI/GI/1/1 and GI/GI/1/2* systems.
%\subsection{Notation}
Consider a single source generating status updates or packets which are immediately dispatched to a single-server queueing system. The inter-arrival time between the packets is denoted by the random variable $Z$ with mean-arrival rate $\lambda = \frac{1}{\mathbb{E}[Z]}$. The arriving packets may be stored in a queue and are served by a server using some scheduling policy. We use the random variable $X$ to denote the service time with mean-service rate $\mu = \frac{1}{\E[X]}$. 
%We index packets in the order of their departure, i.e., if a packet departs $n$th in the sequence of departures, we refer to it as packet $n$. 
We use packet $n$ to refer to a packet that is $n$th in the sequence of departures.
Let $\TD(n)$ denote the time instant of $n$th packet departure and $\TA(n)$ denote the corresponding arrival instant.

The AoI metric, denoted by $\Delta(t)$, is defined as:
\begin{align}\label{eq:AoI-Definition}
\Delta(t) \triangleq t - \max\{\TA(n): \TD(n) \leq t\}. 
\end{align}
For a given age limit $d \geq 0$, we are interested in computing the steady-state violation probability or simply violation probability given by $\P(\Delta > d) = \lim_{t \rightarrow \infty} \P(\Delta(t) > d)$.
\begin{figure}[t]
\centering
\includegraphics[width = 3.2in, height = 1.5in]{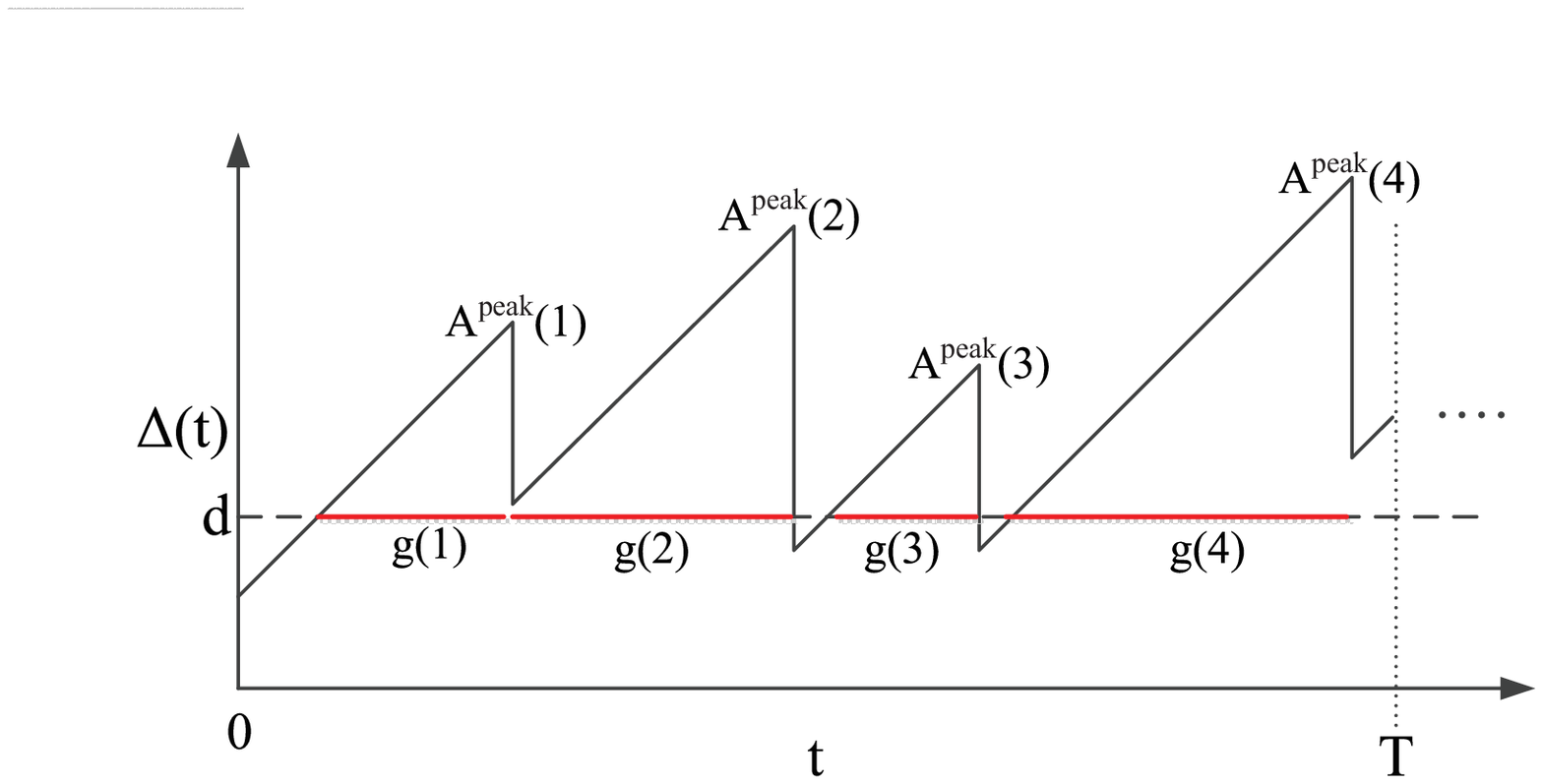}
\vspace{-.2cm}
\caption{A sample path of AoI process.}
\label{fig:AoIprocess}
\vspace{-.4cm}
\end{figure}
%A sample path of $\Delta(t)$ is shown in Figure~\ref{fig:AoIprocess}.

The AoI process increases linearly in time with slope one until the departure of information update packet and it drops to a value equal to the system delay of that packet. Let $\{A^\text{peak}(k), k\geq 1\}$ denote the peak AoI process, where $A^\text{peak}(k)$ denotes the $k$th peak of $\Delta(t)$ as shown in Figure~\ref{fig:AoIprocess}. Let $M(t)$ denote the number of peaks in the interval $(0,t]$. Also, in Figure~\ref{fig:AoIprocess} we plot $g(k)$, which is defined as the time duration for which AoI is greater than an age limit $d$ in the interval between $(k-1)$th peak and $k$th peak. As mentioned before, the characterization of the violation probability and the bounds presented in this paper are obtained in terms of $g(k)$.  

Note that the AoI peaks occur only at the departure instants of packets, but the converse might not be true as some packet departures might not result in a drop in the AoI. This may happen, for instance, in a GI/GI/1 queue under LCFS scheduling. If there is no new arrival during the service of a packet, the next packet from the queue does not reduce AoI upon its departure as its arrival time would be older than that of the previous departure. 
As noted before, we refer to packets that reduce AoI upon their departure as information update packets, and use $k$ to index them as it uniquely identifies an information update packet that departs at $k$th AoI peak. 
%From here on we simply use packet $k$ to refer to $k$th information update packet. 
%With an abuse in notation we use a different indexing system for the packets, where we index a packet by $k$ which refers to an \textit{information update} packet that departs the system at $k$th AoI peak. Accordingly, $\TD(k)$ and $\TA(k)$ are the departure and arrival time instants of a packet that departs at $k$th AoI peak. Under this indexing system, not all packets may get indexed, only those packets whose departures reduce AoI. 
We note that packet $n$ and packet $k$ may not refer to the same packet for $k = n$. Let $\nu = 1/\mathbb{E}[\TD(k) - \TD(k-1)]$ denote the expected departure rate of information update packets.

%\subsection{GI/GI/1/1 and GI/GI/1/2* systems}
%A GI/GI/1/1 system has a server and no queue. An arriving packet is dropped if the server is busy, otherwise it is served immediately. The inter-arrival times are i.i.d. and so are the service times. A GI/GI/1/2* system has a server and unit capacity capacity queue. An arriving packet always replaces the packet in the queue if any or is served immediately if the server is idle. 
We study the GI/GI/1/1 and GI/GI/1/2* systems under non-preemptive scheduling. In both systems the inter-arrival times and the service times are i.i.d. As mentioned before, a packet being served always has arrival time later than that of the previous departure. Thus, AoI is reduced at each departure instant and all departures are information update packets. In these systems, packet $n$ and packet $k$ refer to the same packet for $n = k$, $\TD(k)-\TD(k-1)$ represents the inter-departure time,  $M(t)$ is the number of departures till time $t$, and $\nu$ is simply the expected departure rate.

We use $\omega$ to denote a sample path of AoI, and $\Omega$ to denote the set of all sample paths.
Let $\gamma(k)$ and $\Gamma(k)$ denote the lower and upper bounds for $g(k)$ on any sample path, i.e., 
\begin{align}\label{eq:gkLBUB}
\gamma(\omega,k) \leq g(\omega,k) \leq \Gamma(\omega,k), \, \forall k \text{ and } \forall \omega \in \Omega.
\end{align}
In the rest of the paper, we explicitly drop $\omega$ if it is clear from the context.

\begin{table}[ht]
\renewcommand{\arraystretch}{1.2}
\caption{List of Symbols}
\centering
\begin{tabular}{|l|c|c|}
\hline
$k$ & Index of an information update packet \\
\hline
$\TA(k)$ & Arrival time of packet $k$\\
\hline
$\TD(k)$ & Departure time of packet $k$\\
\hline
$\check{Z}_k$  & Inter-arrival time between packet $k$ and its \textit{previous arrival} \\
\hline
$\hat{Z}_k$  & Inter-arrival time between packet $k$ and its \textit{next arrival} \\
\hline
$X_k$  & Service time of packet $k$ \\
%\hline
%$\lambda, \mu$  & Expected arrival rate, expected service rate \\
\hline
$\nu$  & Expected departure rate of information update packets \\
\hline
$I_k$  & Idle time just before the service of packet $k$ \\
\hline
$W_k$ & Waiting time of packet $k$ \\
\hline
{\color{black} $M(t)$} & {\color{black}Number of AoI peaks in the interval $(0,t]$}\\
\hline
\end{tabular}
\label{tabel1}
\end{table}

The list of symbols used in the paper are summarized in Table~\ref{tabel1}. 
%We use $I_k$ and $W_k$ to denote the idle time before the start of service of packet $k$ and its waiting time, respectively. We use $\check{Z}_k$ to denote the inter-arrival time between packet $k$ and its previous arrival, and $\hat{Z}_k$  to denote the inter-arrival time between packet $k$ and its next arrival.
We use $(x)^+$ for $\max(0,x)$, and $\mathbbm{1}\{\cdot\}$ for the indicator function,  where $\mathbbm{1}\{E\}$ equals one if event $E$ is true, and is zero, otherwise. Finally, we use the functions $F_Y(\cdot)$ and $f_Y(\cdot)$ to denote the cumulative distribution function and the probability density function of a random variable $Y,$ respectively.

\section{General Results: Distribution of AoI and Bounds}\label{sec:generalresults}
In this section, we provide expressions for the distribution of AoI as well as bounds on the age limit violation probability. These results are general in the sense that they are applicable to any single-source-single-server queueing system. We first obtain a general characterization for the violation probability in terms of $g(k)$, assuming the process is stationary and ergodic.
We then present an upper bound characterization for systems where the violation probability characterised above can not be directly computed. Finally, we establish a generic lower bound that is used to analyse the performance of the proposed upper bounds.  

\subsection{AoI: Fundamental Relations}
Recall that AoI process increases linearly with slope one until the next information packet departure. Therefore, $A^\text{peak}(k)$ can be determined from $\Delta(t)$ at departure of packet $(k-1)$,  and the inter-departure time between $k$th and $(k-1)$th packets. 
{\allowdisplaybreaks \begin{align}\label{eq:Apeakdef}
A^\text{peak}(k) &= \TD(k) - \TA(k-1) \nonumber  \\
&= \Delta(\TD(k-1)) + \TD(k) - \TD(k-1).
\end{align}}
Note that, in general, there can be packet departures in between $(k-1)$th and $k$th packets.

%Let $g(k)$ denote the time duration for which AoI is greater than an age limit $d$ during its $k$th peak, i.e., in the interval $[\TD(k-1),\TD(k))$. We depict $g(k)$ in Figure~\ref{fig:AoIprocess}. 
Analysing $g(k)$ is central to the results presented in this paper. In the following lemma, we express $g(k)$ in terms of $k$th  AoI peak and inter-departure time between $k$th and $(k-1)$th packet.
\begin{lemma}\label{lem:gk}
Given $d \geq 0$, for any sample path of $\Delta(t)$,
\begin{align}\label{eq:g(k)}
g(k) =\min\{(A^\text{peak}(k)-d)^+,\TD(k) - \TD(k-1)\}, \, \forall k.
\end{align}
\end{lemma}
\begin{proof}
Consider the case where $A^\text{peak}(k) \leq d$. For this case $g(k)$ is zero, by definition, which is satisfied by \eqref{eq:g(k)} as $\TD(k) \geq \TD(k-1)$. For $A^\text{peak}(k) > d$ we further consider the following cases.

\textbf{Case 1:} $A^\text{peak}(k) > d$ and $A^\text{peak}(k) - d > \TD(k) - \TD(k-1)$. Using this in~\eqref{eq:Apeakdef}, we obtain $\Delta(\TD(k-1)) > d$. This implies that $\Delta(t) > d$ during the entire interval $[\TD(k-1),\TD(k))$. Therefore, $g(k)= \TD(k) - \TD(k-1)$. This is the case for $g(2)$ in Figure~\ref{fig:AoIprocess}.

\textbf{Case 2:} $A^\text{peak}(k) > d$ and $A^\text{peak}(k) - d \leq \TD(k) - \TD(k-1)$. Using this in~\eqref{eq:Apeakdef}, we obtain $\Delta(\TD(k-1)) \leq d$. In this case the horizontal line, with $y$ coordinate equal to $d$, intersects $\Delta(t)$ at some time $t' \in [\TD(k-1),\TD(k))$. Since $\Delta(t)$ increases linearly with slope one, by geometry we obtain $g(k) = \TD(k) - t' = A^\text{peak}(k) - d$.

From the above analysis, we conclude that $g(k)$ takes the minimum value of $(A^\text{peak}(k) - d)^+$ and $\TD(k) - \TD(k-1)$, and the lemma follows. 
%the fraction of time $\Delta(t)$ is greater than $d$ is given by $(\sum_{k=1}^{M(T)}g(k) + \delta(T))/T$, where $\delta(T)$ is the duration for which $\Delta(t)$ is greater than $t$ after the peak $M(T)$ and before time $T$.
%Since the AoI process is stationary and ergodic, the fraction of time $\Delta(t)$ is greater $d$ as $t$ goes to infinity converges to violation probability, almost surely (). Therefore, from the above analysis we have
\end{proof}

Next, we characterize the violation probability in terms of $g(k)$ in the following theorem.
\begin{theorem}\label{thm:vioprobAoIPeak}
If the AoI process is stationary and ergodic, given $d \geq 0$, the AoI violation probability, if exists, is given by
\begin{align}\label{eq:vioprobAoIPeak}
\P(\Delta > d) = \lim_{T \rightarrow \infty} \frac{1}{T} \sum_{k = 1}^{M(T)} g(k), \text{ a.s.},
\end{align}
where $g(k)$ is given in~\eqref{eq:g(k)}.
\end{theorem}
\begin{proof}
Since $\Delta(t)$ is stationary and ergodic, by Birkhoff's ergodic theorem~\cite{Kumar_2004}, we have
\begin{align}\label{eq:VioProbIndicator}
\P(\Delta > d) = \lim_{T \rightarrow \infty} \frac{1}{T}\int_{0}^{T} \mathbbm{1}{\{\Delta(\tau) > d\}} d\tau, \text{ a.s.}
\end{align}
The RHS above is the fraction of time $\Delta(t)$ is greater than $d$ in a given sample path. 

Consider a sample path of $\Delta(t)$ presented in Figure~\ref{fig:AoIprocess}. Let $\delta(T)$ denote the duration for which $\Delta(t)$ is greater than $d$ after the $M(T)$th peak and before time $T$. It is easy to see that
\begin{align}\label{eq:Ideltak}
\int_{0}^{T} \mathbbm{1}{\{\Delta(\tau) > d\}}d \tau &= \sum_{k=1}^{M(T)}g(k) + \delta(T) \nonumber\\
\Rightarrow \lim_{T \rightarrow \infty} \frac{1}{T}\int_{0}^{T} \mathbbm{1}{\{\Delta(\tau) > d\}} d\tau &= \lim_{T \rightarrow \infty} \frac{1}{T}\sum_{k=1}^{M(T)}g(k).
\end{align}
In the last step above we have used the fact that  $\frac{\delta(T)}{T}$ goes to zero as $T$ goes to infinity. The result follows by substituting~\eqref{eq:Ideltak} in~\eqref{eq:VioProbIndicator}. 
\end{proof}
%A more rigorous proof for Theorem~\ref{thm:vioprobAoIPeak} using Birkhoff's Ergodic Theorem involves defining an appropriate real-valued function and a measure-preserving transformation on $\Delta(t)$ and is not presented here.
%f({X1,X2,...,X_n}) = \mathbbm{1}{X_n > d}, T^j \Delta(t) = \Delta(t_j), where t_j is end time of j-th slot when  the time axis is slotted by duration \delta.
Theorem~\ref{thm:vioprobAoIPeak} is quite general in the sense that it holds for any scheduling policy (e.g., FCFS/LCFS, preemptive/non-premptive etc.), general service times (possibly correlated), and general inter-arrival times (possibly correlated), as long as it is ensured that the resulting AoI process is stationary and ergodic. Note that even if the AoI process is stationary and ergodic, the violation probability may not exist. For example, for a D/G/1 system using FCFS the violation probability does not exist if $d < \frac{1}{\lambda}$~\cite{Champati_2018a}. 

As mentioned before, a general formula was proposed in~\cite{Yoshiaki2018} that characterizes the distribution of AoI in terms of the distribution of peak AoI, and the distribution of system delay. In contrast to the general formula, Theorem~\ref{thm:vioprobAoIPeak} characterizes the violation probability in terms of $g(k)$, which is a function of peak AoI process and the inter-departure time between information update packets during the $k$th AoI peak. As we will show later, it is easy to obtain upper bounds and lower bounds for $g(k)$ which enables us to derive upper bounds for the violation probability and analyse their worst-case performance in cases where an exact expression is not tractable.
%We note however that this may not be possible using the general formula in~\cite{Yoshiaki2018}, as it requires the distribution of peak AoI and the distribution of system delay. 
%this characterization will not only allows us in obtaining exact expressions but also results in a simple method for deriving and analysing upper bounds for violation probability. 
%Apart from its use for analytical purposes, Theorem~\ref{thm:vioprobAoIPeak} eases the computation of the violation probability in simulation, wherein we directly compute $g(k)$ at the departure instants.

%In~\cite{Yoshiaki:2017}, the authors provided a general result for AoI sample paths that relates the distribution of the AoI with the distribution of peak AoI and the distribution of a discrete process that samples AoI process at departure instants. The result in Theorem~\ref{thm:vioprobAoIPeak} is in spirit close to their result, but differs in the sense it relates distribution of AoI with peak AoI process, rather than its distribution, and the inter-departure process to the distribution of AoI. However, they provide closed-form expressions for the LST of the AoI distribution for the cases M/GI/1 and GI/M/1 - require the residual distribtuion of the service

The challenge in evaluating the infinite summation in the RHS of~\eqref{eq:vioprobAoIPeak} is that the sequence $\{g(k),k \geq 1\}$ is not i.i.d., and we cannot directly use the Strong Law of Large Numbers (SLLN). However, we will later show that quantities involving $g(k)$ have structural independence property, defined below, which enables us to use SLLN.

\begin{definition}\label{def:siid}
An infinite sequence of random variables $\{X_n, n\geq 1 \}$ is structurally independent and identically distributed (s.i.i.d.) iff $X_n$ are identically distributed and have the following structural independence: for $1 \leq m < \infty$, $X_{i+jm}$ is independent of $X_{i+km}$, for all $1\leq i \leq m$, $j \geq 0$, $k \geq 0$, and $j \neq k$.
\end{definition}

In the results that follow we make use of the following lemma, which extends SLLN for s.i.i.d. random variables.
\begin{lemma}\label{lem:SLLN}
For any sequence $\{X_n, n\geq 1 \}$ that is s.i.i.d. according to Definition 1, we have
\begin{align*}
\lim_{N \rightarrow \infty} \frac{1}{N} \sum_{n=0}^{N} X_{n} = \mathbb{E}[X], \text{ a.s.}, 
\end{align*}
where $\mathbb{E}[X] = \mathbb{E}[X_n]$ for all $n$. 
\end{lemma}
\begin{proof}
The proof is based on partitioning the sum into multiple terms which themselves are infinite sums of  i.i.d. random variables and then apply SLLN for these summations.
{\allowdisplaybreaks
\begin{align*}
&\lim_{N \rightarrow \infty} \frac{1}{N} \sum_{n=0}^{N} X_{n} = \lim_{N \rightarrow \infty} \frac{1}{N}  \sum_{i=1}^{m} \sum_{j=1}^{\lfloor\frac{N-i+m }{m}\rfloor} X_{i+(j-1)m} \\
&= \frac{1}{m}\sum_{i=1}^{m} \lim_{N \rightarrow \infty} \frac{\lfloor\frac{N-i+m }{m}\rfloor}{N/m}\sum_{j=1}^{\lfloor\frac{N-i+m }{m}\rfloor} \frac{X_{i+(j-1)m}}{\lfloor\frac{N-i+m }{m}\rfloor} \\
&= \frac{1}{m}\sum_{i=1}^{m} \mathbb{E}[X] = \mathbb{E}[X], \text{ a.s.}
\end{align*}}
In the third step above, we have used SLLN as $\{X_{i+(j-1)m},j \geq 1\}$ are i.i.d. (Definition~\ref{def:siid}), and $\lfloor\frac{N-i+m }{m}\rfloor$ differs from $\frac{N}{m}$ by utmost $1$.
\end{proof}

%Let $\nu$ denote the expected departure rate, i.e.,
%\begin{align}
%\nu = \frac{1}{\mathbb{E}[\TD(k)-\TD(k-1)]}.
%\end{align}
%In the following theorem we present 
\begin{theorem}\label{thm:evaluategk}
Given age limit $d \geq 0$, $\lambda > 0$, $0 < \mathbb{E}[X] = \frac{1}{\mu} < \infty$,
 $\{g(k),k \geq 1\}$ are s.i.i.d., and $\{\TD(k)-\TD(k-1),k \geq 1\}$ are s.i.i.d., then
\begin{align*}
\lim_{T \rightarrow \infty} \frac{1}{T} \sum_{k = 1}^{M(T)} g(k) = \nu\mathbb{E}[g(k)], \text{ a.s.}
\end{align*}
\end{theorem}
\begin{proof} 
We have
\begin{align}\label{eq1:thm3}
\lim_{T \rightarrow \infty} \frac{1}{T} \sum_{k = 1}^{M(T)} g(k) = \lim_{T \rightarrow \infty} \frac{M(T)}{T}\cdot \frac{\sum_{k = 1}^{M(T)} g(k)}{M(T)}.
\end{align}
%In the following we evaluate the departure rate $\frac{M(T)}{T}$, for $T$ approaching infinity. We have 
%\begin{align*}
%T = \sum_{k=1}^{M(T)}(\TD(k) - \TD(k-1)) + \delta(T),
%\end{align*}
%where $\delta(T) = T - \TD(M(T))$. 
%Since $\lambda > 0$ and $\mathbb{E}[X] < \infty$, $M(T)$ approaches infinity and $\frac{\delta(T)}{M(T)}$ approaches zero, almost sure, as $T$ approaches infinity. 
Since $\lambda > 0$ and $\mathbb{E}[X] < \infty$, $M(T)$ approaches infinity, almost surely, as $T$ approaches infinity, and we obtain,
\begin{align*}
\lim_{T\rightarrow \infty} \frac{T}{M(T)} = \lim_{T\rightarrow \infty} \sum_{k=1}^{M(T)}(\TD(k) - \TD(k-1))/M(T).
\end{align*}

Since $\{\TD(k) - \TD(k-1),k \geq 1\}$ are s.i.i.d., from Lemma~\ref{lem:SLLN} we have
\begin{align}\label{eq3:interdepart}
\lim_{T\rightarrow \infty} \frac{T}{M(T)} &= \mathbb{E}[\TD(k) - \TD(k-1)], \text{ a.s.}
\end{align}
Similarly, we invoke Lemma~\ref{lem:SLLN} for $\{g(k),k \geq 1\}$ and obtain  
\begin{align}\label{eq:ExpYk}
&\lim_{M(T) \rightarrow \infty} \frac{\sum_{k = 1}^{M(T)} g(k)}{M(T)} = \mathbb{E}[g(k)], \text{ a.s.}
\end{align}
The result follows by substituting~\eqref{eq3:interdepart} and~\eqref{eq:ExpYk} in~\eqref{eq1:thm3}.
\end{proof}
Theorem~\ref{thm:evaluategk} can be seen as an extension of renewal reward theorem for s.i.i.d. renewals and rewards. Later, we use the theorem to derive exact expressions for the violation probability for the D/GI/1/1 and M/GI/1/1 systems.

\subsection{Bounds for AoI Violation Probability}
%We note that proving the ergodicity of the AoI process might be challenging, in general. Even if we prove that the AoI process is ergodic, we may not be able to obtain closed-form expression for the RHS of~\eqref{eq:vioprobAoIPeak}. 
%Even if $\{g(k),k \geq 1\}$ are s.i.i.d., and $\{\TD(k)-\TD(k-1),k > 1\}$ are s.i.i.d., 
As one can expect $g(k)$ and $\TD(k)-\TD(k-1)$ depend on the idle time $I_k$ and waiting time $W_k$ in the queuing system. Therefore, computing $\mathbb{E}[g(k)]$ and $\nu$ is hard, in general, as the distributions of $I_k$ and $W_k$ become intractable for general inter-arrival-time and service-time distributions.
To this end, in the following theorem we present a result that is useful in deriving upper bounds for the violation probability and only requires the AoI process to be stationary.

%To be precise, $\gamma(k)$ and $\Gamma(k)$ are lower and upper bounds for any sample path $\omega$, i.e., $\gamma(\omega,k) \leq g(\omega,k) \leq \Gamma(\omega,k)$, for all $k$ and for all $\omega$.
%\begin{align*}
%\lim_{T \rightarrow \infty} \frac{1}{T} \sum_{k = 1}^{M(T)} g(k) \leq \beta,
%\end{align*}
%\P(\Delta > d)
\begin{theorem}\label{thm:vioprobAoIPeakUB}
If the AoI process is stationary, then 
{\allowdisplaybreaks \begin{align*}
\mathbb{E}_\omega \!\!\left[\!\lim_{T \rightarrow \infty} \! \frac{1}{T} \!\! \sum_{k = 1}^{M(T)} \!\!\!\gamma(\omega,\!k)\! \right]\!\! \leq \! \P(\Delta \! > \! d)  \!\leq \! \mathbb{E}_\omega \!\! \left[\lim_{T \rightarrow \infty}\! \frac{1}{T} \!\! \sum_{k = 1}^{M(T)}\!\!\! \Gamma(\omega,\!k) \!\right]\!.
\end{align*}}
\end{theorem}
\begin{proof}
Since $\Delta(t)$ is stationary, we have
\begin{align*}
\P(\Delta(t) > d) &= \mathbb{E}_\omega[ \mathbbm{1}{\{\Delta(\omega,t) > d\}}], \, \forall t.
\end{align*}
Therefore, for any $t$,
\begin{align}\label{eq:gomegak}
\P(\Delta(t) > d) &= \lim_{T \rightarrow \infty} \frac{1}{T}\int_{0}^{T} \mathbb{E}_\omega[ \mathbbm{1}{\{\Delta(\omega,t) > d\}}] dt \nonumber \\
&= \mathbb{E}_\omega\left[\lim_{T \rightarrow \infty} \frac{1}{T}\int_{0}^{T} \mathbbm{1}{\{\Delta(\omega,t) > d\}} dt \right] \nonumber \\
&= \mathbb{E}_\omega\left[\lim_{T \rightarrow \infty} \frac{1}{T} \sum_{k = 1}^{M(T)} g(\omega,k) \right].
\end{align}
Second step above is due to the fact that indicator function is non-negative. The third step is due to the fact that~\eqref{eq:Ideltak} is true for any $\omega$. The result follows from~\eqref{eq:gomegak} and~\eqref{eq:gkLBUB}.
\end{proof}
In terms of applicability, Theorem~\ref{thm:vioprobAoIPeakUB} is more general than Theorem~\ref{thm:vioprobAoIPeak} as it does not require ergodicity of the AoI process.
Following Theorem~\ref{thm:vioprobAoIPeakUB}, we strive to obtain upper bounds for the violation probability for GI/GI/1/1 and GI/GI/1/2* systems by finding bounds for $g(k)$.

In the following we establish a lower bound for $g(k)$ that is applicable to any single-source-single-server queueing system. 
\begin{lemma}\label{lem:gkLB}
For a single-source-single-server queuing system, it is true that $g(k) \geq \gamma^*(k)$, for all $k$, where 
%\begin{align*}
%\lim_{t \rightarrow \infty} P(\Delta(t) > d) \geq \frac{\mathbb{E}[\gamma^*(k)]}{\mathbb{E}[\TD(k) - \TD(k-1)]},
%\end{align*}
\begin{align*}
\gamma^*(k) = \min\{(X_k + X_{k-1} + I_k - d)^+,X_k + I_k\}.
\end{align*}
\end{lemma}
\begin{proof}
For a single-server system it is easy to see that the inter-departure time between information update packets is at least the service time of a packet and idle time before its service started, i.e.,
\begin{align}\label{eq:lb1}
\TD(k) - \TD(k-1) \geq X_k + I_k.
\end{align}
From~\eqref{eq:Apeakdef} we have  
\begin{align*}
A^\text{peak}(k) &= \TD(k) - \TA(k-1)\\
&\ge \TD(k) - (\TD(k-1) - X_{k-1} )\\
&\ge  X_k + I_k + X_{k-1}
\end{align*}
The second step is due to the fact that a packet departure time is at least equal to its arrival time plus its service time. The last step is due to~\eqref{eq:lb1}.
%We claim that $A^\text{peak}(k)$ is lower bounded as follows.
%\begin{align}\label{eq:lb2}
%A^\text{peak}(k) \geq X_k + X_{k-1} + I_k.
%\end{align}
%To see this, note that $A^\text{peak}(k) = \TD(k) - \TA(k-1)$, which includes the service times of  packets $(k-1)$ and $k$ and the idle time before packet $k$. The result follows by substituting~\eqref{eq:lb1} and~\eqref{eq:lb2} in~\eqref{eq:g(k)}.
\end{proof}
%Note that $\gamma^*(k)$ are s.i.i.d., as $\gamma^*(k+2)$ is independent of the random variables $\{\gamma(i), i \leq k\}$.  
In this paper, we use the lower bound in Lemma~\ref{lem:gkLB} to analyse the performance of the upper bounds derived for the AoI violation probability for GI/GI/1/1 and GI/GI/1/2* systems. Nevertheless, this method is quite general and can be applied to other queueing systems.  
%Note that for specific queuing systems we may find tighter lower bounds than the one Lemma~\ref{lem:gkLB}.

\section{The GI/GI/1/1 System}\label{sec:GG11}
\begin{comment}
Let $Z_k$ denote the inter-arrival time between packet $k$ and its previous arrival. Note that the previous arrival may not be packet $k-1$ as there can be arrivals while packet $k-1$ is being served. Of course $Z_k$ has the same distribution as the inter-arrival time which we denote by $F_Z$, and the corresponding probability density function is $f_Z$. The mean inter-arrival time $\mathbb{E}[Z] = \frac{1}{\lambda}$.

%Let $a_k$ denote the number of arrivals, including packet $k$, during the inter-arrival time between packets $k$ and $k-1$. Consider the arrivals that occurred during this inter-arrival time. With some abuse in notation, we denote the inter-arrival times for these arrivals by $Z^{k}_i$ for $1 \leq i \leq a_{k}$. By definition of $a_k$, we have
Since the service of packet $k-1$ starts up on its arrival we have
\begin{align*}
X_{k-1} \geq \sum_{i=1}^{a_{k-1}}Z^{k-1}_i.
\end{align*}
Note that the $a_{k-1}$th packet is packet $k$. Also, we have $Z^{k-1}_{a_{k-1}} = Z_k$, by definition of $Z_k$. Therefore, the time difference between the arrival of packet $k$ and packet $k-1$ is given by 
\begin{align}\label{eq:inter-arr}
\TA(k) - \TA(k-1) = \sum_{i=1}^{a_{k}}Z^{k}_i
\end{align}
Under the indexing system that we use all the packets that are discarded in a GI/GI/1/1 system do not get indexed. In particular, a packet $k$ is always served upon its arrival and all packets that are discarded while packet $k$ is served are not indexed and the arrival that occurs immediately after the departure of packet $k$ is indexed $(k+1)$. 
\end{comment}

In this section, we present a general characterization for the violation probability for the GI/GI/1/1 system. For D/GI/1/1 and M/GI/1/1 we obtain exact expressions for the violation probability and consequently derive the same for a system using zero-wait policy. Finally, we provide an upper bound for the violation probability and analyse its worst-case performance.
%\footnote{The details of the proofs for some of the results provided below are removed from this manuscript due to lack of space. We did this mainly for the results where a reader may, with some effort, reproduce these proofs.}
%We did this only for results where a reader may, with some effort, reproduce these proofs.

\begin{comment}
The idle time $I_k$ is then given by
before the service of packet $k$ is equal to the difference between the arrival time of packet $k$ and the departure time of packet $k-1$, i.e.,
\begin{align}\label{eq:idletimeGG11}
I_k = \TA(k)\! -\! \TD(k-1) = \TA(k)\! -\! \TA(k-1)\! -\! X_{k-1}.
\end{align} 
% &= \TA(k) - \TA(k-1) + \TA(k-1) - \TD(k-1) \nonumber \\
We note that $a_k$ is a stopping time, i.e., it only depends $\{Z^k_i, 1 \leq i \leq a_k\} \cup \{X_{k-1}\}$, but not on future $Z^k_{i}$. In other words, by observing the inter-arrival times so far, since the arrival of packet $k-1$, whether to wait for the next arrival or stop by labelling the current arrival as packet $n$. Using Wald's lemma we obtain the expected idle time as follows.
\begin{align*}
\mathbb{E}[I_k] = \mathbb{E}[a_k]\mathbb{E}[Z].
\end{align*}
It is hard to derive $\mathbb{E}[I_k]$ for a GI/GI/1/1 system, in general, as we need to compute $\mathbb{E}[a_k]$. However, for D/G/1/1 and M/G/1/1 systems it can be computed.
\end{comment}

In a GI/GI/1/1 system, packet $k$ is served upon its arrival, which implies $\TD(k) = \TA (k) + X_k$. Further, the inter-departure time is given by $\TD(k) - \TD(k-1) = X_k + I_k$.
We note that this relation is equally valid for the GI/GI/1/2* system. Therefore, for both systems 
\begin{align}\label{eq:nuGG112}
\nu = 1/(\mathbb{E}[X_k] + \mathbb{E}[I_k]).
\end{align}
In the following we compute $A^\text{peak}(k)$ for a GI/GI/1/1 system. 
\begin{align}\label{eq:Apeak}
&A^\text{peak}(k) = \TD(k) - \TA(k-1) \nonumber\\
&= \TD(k)\! - \!\TA(k)\! +\! \TA(k)\! -\! \TD(k\!-\!1)\! +\! \TD(k\!-\!1)\! -\!  \TA(k\!-\!1) \nonumber\\
&= X_{k} + I_k + X_{k-1}.
\end{align}
The following lemma immediately follows from the above analysis and Lemma~\ref{lem:gk}.
\begin{lemma}\label{lem:gkGG11}
In a GI/GI/1/1 system, given $d \geq 0$, for any sample path of $\Delta(t)$ the corresponding g(k) is given by 
\begin{align}\label{eq:gkGG11}
g(k) = \min\left\{(X_{k-1} + I_k + X_k -d)^+,X_k + I_k\right\}, \forall k
\end{align}
\end{lemma}
%\begin{proof}
%We claim that
%
%Since packet $k$ is served on its arrival, we have $\TD(k) = \TA (k) + X_k$.
%Using this in~\eqref{eq:TD}, we obtain
%\begin{align}\label{eq1:Apeak}
%A^\text{peak}(k) &= \Delta(\TD(k-1)) + \TA (k) + X_k  - \TD(k-1) \nonumber\\
%&= - \TA(k-1) + \TA (k) + X_k.
%\end{align}
%Our claim can be verified by~\eqref{eq:idletimeGG11} in~\eqref{eq1:Apeak}. The result follows by substituting~\eqref{eq:interDepart} and~\eqref{eq:Apeak} in the definition of $g(k)$ (Theorem 1). 
%\end{proof}

We now provide a general expression for the violation probability in the following theorem. 
\begin{theorem}\label{thm:GG11}
Consider a GI/GI/1/1 system, assuming the AoI process is stationary and ergodic, then for all $d \geq 0$, $\lambda > 0$, and $0 < \mathbb{E}[X] = \frac{1}{\mu} < \infty$, the violation probability, if exists, is given by:
\begin{align*}
\P(\Delta > d) = \nu \mathbb{E}[g(k)], \, \text{ a.s.},
\end{align*}
where $g(k)$ is given by~\eqref{eq:gkGG11} and $\nu$ is given by~\eqref{eq:nuGG112}.
\end{theorem}
\begin{proof}
We note that the inter-arrival times $\{\TA(k)-\TA(k-1),k \geq 1\}$ in a GI/GI/1/1 system are i.i.d. To see this, the duration $\TA(k)-\TA(k-1)$ equals the sum of inter-arrival times of all dropped packets and the packet $k$ starting from packet $k-1$, and only depends on the inter-arrival time $Z$ and the service time of packet $k-1$. Therefore, the start of service of a packet is a renewal instant. This implies $I_k$ are i.i.d. which further implies that $\TD(k) - \TD(k-1)$ are i.i.d.
%Recall that $\TD(k) - \TD(k-1) = X_k + I_k$. Note that $\TD(k) - \TD(k-1)$ are identically distributed random variables. From~\eqref{eq:idletimeGG11}, we note that $I_k$ is independent of inter-arrival times and service times prior to packet $(k-1)$. Thus, the inter-departure time $\TD(k) - \TD(k-1)$ is independent of inter-departure times prior to packet $(k-2)$. Therefore, the sequence $\{\TD(k) - \TD(k-1),k\geq 1\}$ satisfies Definition~\ref{def:siid}. 
%In the last step above we have used~\eqref{eq2:interdepart} and $\mathbb{E}[X_{k-1}] = \mathbb{E}[X_{k}]$.
%Note that $\frac{K(T)}{T}$ denotes the departure rate. For $R < \mu$, in the steady-state the departure rate converges to $R$. If $R \geq \mu$, the departure rate converges to the service rate $\mu$. Therefore,
%\begin{align}\label{eq:depRate}
%\lim_{T \rightarrow \infty} \frac{K(T)}{T} = \min\{R,\mu\} \; \text{a.s.}
%\end{align}
From~\eqref{eq:gkGG11} we infer that $g(k)$ are identically distributed random variables, and  $g(k+2)$ is independent of the random variables $\{g(n),1\leq n \leq k\}$ for all $k$. Therefore, the sequence $\{g(k),k \geq 1\}$ is s.i.i.d. 
The result then follows from Theorems~\ref{thm:vioprobAoIPeak} and~\ref{thm:evaluategk}.
\end{proof}
{\color{black} Note that to compute the violation probability, we must compute $\mathbb{E}[g(k)]$. In the derivations that follow, we first compute the distribution of $g(k)$ toward this purpose. The following lemma presents a simplified expression for the distribution of $g(k)$. 
\begin{lemma}\label{lem:g(k)Distr}
For a GI/GI/1/1 system,
\begin{align*}
\P(g(k) > y) &= \int_{0}^{d} \P(X_k + I_k > y-x+d)f_X(x)dx\\
&\quad\quad\quad  + \int_{d}^{\infty}\P(X_k+I_k>y)f_X(x) dx
\end{align*}
\end{lemma}
\begin{proof}
The proof is given in Appendix~\ref{proof:lem:g(k)Distr}.
\end{proof}}

\subsubsection*{Zero-wait policy} In a single-source-single-server queueing system using zero-wait policy, the source generates a packet only when there is a departure. It is easy to see that the statistics of the AoI process for this system will be same as that of GI/GI/1/1 when the input rate approaches infinity. Therefore, the following corollary immediately follows from Theorem~\ref{thm:GG11}, by substituting $I_k=0$ as input rate is infinity. 
\begin{corollary}
For the system using zero-wait policy, the violation probability is given by $\nu \mathbb{E}[g(k)]$, almost surely, where $g(k) = \min\{(X_{k-1}+X_k-d)^+,X_k\}$ and $\nu = \mu$.
%\begin{align*}
%\P(\Delta > d) = \frac{\mathbb{E}[\min\{(X_{k-1}+X_k-d)^+,X_k\}]}{\mathbb{E}[X_k]}, \text{ a.s.}
%\end{align*}
\end{corollary}  

{\color{black} Since the AoI process is non-negative, the expected AoI for zero-wait policy is given by
\begin{align*}
\mathbb{E}[\Delta(t)] =  \int_{0}^{\infty}\nu \mathbb{E}[\min\{(X_{k-1}+X_k-y)^+,X_k\}]dy.
\end{align*}}

Next, we derive exact expressions for AoI violation probability for the D/GI/1/1 and M/GI/1/1 systems.%\footnote{The exact expressions for D/G/1/1 and M/G/1/1 can be equivalently computed starting with the general formula provided in [Theorem 7,~\cite{Yoshiaki2018}].}.
\subsection{D/GI/1/1: Exact Expressions}
In a D/GI/1/1 system, the inter-arrival is deterministic and is equal to $\frac{1}{\lambda}$. To the best of our knowledge, for this system no results exists even for the expected AoI. 

Intuitively, in a D/GI/1/1 system, we only need to consider the rate region $\lambda \geq \frac{1}{d}$ as AoI cannot be less than $\frac{1}{\lambda}$ when the samples are generated at rate $\lambda$. The following lemma asserts this intuition.
\begin{lemma}\label{lem:existenceDG11}
For the D/GI/1/1 system, given $d \geq 0$ and $\lambda > 0$, the AoI violation probability only exists for $d\geq \frac{1}{\lambda}$.
\end{lemma}
\begin{proof}
%The proof follows similar steps as in proof of [Theorem 1~\cite{Champati_2018a}] and is omitted.
The proof is given in Appendix~\ref{proof:lem:existenceDG11}.
\end{proof}

We now present a closed form expression for the violation probability in the following theorem.
\begin{theorem}\label{thm:DG11UB}
For a D/GI/1/1 system, given $d \geq \frac{1}{\lambda}$, $\lambda > 0$, and $0 < \mathbb{E}[X] = \frac{1}{\mu} < \infty$, the violation probability is given by $\nu\mathbb{E}[g(k)]$, almost surely,
%\begin{align*}
%\P(\Delta > d) = \nu\mathbb{E}[g(k)] \, \text{ a.s.},
%\end{align*}
where $g(k)$ is given by~\eqref{eq:gkGG11}, $\nu = \lambda/\mathbb{E}[\lceil \lambda X_k \rceil]$ and $I_k = \lceil \lambda X_{k-1} \rceil/\lambda - X_{k-1}$.
\end{theorem}
\begin{proof}
%We have $\mathbb{E}[Z] = \frac{1}{\lambda}$, and $\mathbb{E}[X_k] = \mathbb{E}[X_{k-1}] = \mathbb{E}[X]$. 
Using the results from Lemma~\ref{lem:gkGG11} and Theorem~\ref{thm:GG11}, it is sufficient to show that $I_k = \frac{\lceil \lambda X_{k-1} \rceil}{\lambda} - X_{k-1}$, which we argue to be true in the following. The time difference between the arrival of packet $k$ and packet $(k-1)$ is given by $\frac{\lceil \lambda X_{k-1} \rceil}{\lambda}$. To see this, the service of packet $k-1$ starts upon its arrival, i.e., at $\TA(k-1)$. During the service of packet $k-1$ the packets that arrived would be dropped and the packet that arrived immediately after $\TA(k-1) + X_{k-1}$ is served. The number of arrivals since $\TA(k-1)$ is given by $\lceil \lambda X_{k-1} \rceil$, and the time elapsed is $\frac{\lceil \lambda X_{k-1} \rceil}{\lambda}$. This implies that the idle time $I_k$  is given by $\frac{\lceil \lambda X_{k-1} \rceil}{\lambda} - X_{k-1}$. 
\end{proof}

In the following we compute the expression provided in Theorem~\ref{thm:DG11UB} for exponential-service-time distribution.
\begin{corollary}\label{cor:DM11}
For a D/M/1/1 queue, given $d \geq \frac{1}{\lambda}$, $\lambda > 0$, and $0 < \mathbb{E}[X] = \frac{1}{\mu} < \infty$, the violation probability is given by $\nu \mathbb{E}[g(k)]$, almost surely, where $\nu = \lambda(1-e^{-\mu/\lambda})$ and
\begin{align*}
\mathbb{E}[g(k)] &= \frac{e^{-\mu \ceild}}{\lambda(1-e^{-\frac{\mu}{\lambda}})} + e^{-\mu \frac{\lfloor \lambda d\rfloor}{\lambda}}\left[\ceild-d + \frac{1}{\mu}\right] \\
&\quad + \frac{e^{-\mu d}}{\mu}\left((e^{\frac{\mu}{\lambda}}-1)\lfloor \lambda d\rfloor - 1\right).
\end{align*}
%\begin{align*}
%\mathbb{E}[\lceil \lambda X \rceil] &= 1/(1-e^{-\mu/\lambda}).
%\end{align*}
\end{corollary}
\begin{proof}
%The proof is omitted due to space limitation.
{\color{black} The proof is given in Appendix~\ref{proof:cor:DM11}.}
%The proof is given in technical report~\cite{techreport}.
\end{proof}

\begin{comment}
\begin{figure}
\centering
\includegraphics[width = 2.7in]{20180706_VioProbDM11.eps}
\caption{Violation probability vs the arrival rate for an D/M/1/1 system for different age limits $d$ and $\mu = 1$. }
\label{fig:20180706_VioProbDM11}
\vspace{-.5cm}
\end{figure}
\end{comment}

\subsection{M/GI/1/1: Exact Expressions}
For M/GI/1/1 system, the authors in~\cite{Najm_2017} derived expressions for the expected AoI and the expected peak AoI. For this system we provide an expression for the violation probability of AoI. 
\begin{theorem}\label{thm:MG11}
For an M/GI/1/1 system, given $d \geq 0$, $\lambda > 0$, and $0 < \mathbb{E}[X] = \frac{1}{\mu} < \infty$, the violation probability, if exists, is given by $\nu\mathbb{E}[g(k)]$, almost surely, where $g(k)$ is given in~\eqref{eq:gkGG11}, $\frac{1}{\nu} = \frac{1}{\lambda} + \frac{1}{\mu}$, and $I_k \sim \text{Exp}(\lambda)$. 
\end{theorem}
\begin{proof}
The result follows from Theorem~\ref{thm:GG11} and using the fact that in an M/G/1/1 system $I_k$ and the inter-arrival times are identically distributed.
\end{proof}

For the special case of M/M/1/1, we have the following corollary.
\begin{corollary}\label{cor:MM11}
For the M/M/1/1 system, given $d \geq 0$, $\lambda > 0$, and $0 < \mathbb{E}[X] = \frac{1}{\mu} < \infty$, the violation probability, if exists, is given by $\nu\mathbb{E}[g(k)]$, almost surely, where $\frac{1}{\nu} = \frac{1}{\lambda} + \frac{1}{\mu}$, and
\begin{eqnarray*}
\mathbb{E}[g(k)] \! = \! \left\{\begin{array}{lc}
	 \frac{\mu^2(e^{-\lambda d}-e^{-\mu d})}{\lambda(\mu-\lambda)^2} +e^{-\mu d}\left(\frac{1}{\lambda} + \frac{1}{\mu} - \frac{\lambda d}{\mu - \lambda}\right) & \lambda \neq \mu,\\
   \frac{\mu e^{-\mu d}}{2}\left(d+\frac{2}{\mu}\right)^2 & \lambda = \mu.
  \end{array}\right.
\end{eqnarray*}
\end{corollary} 
\begin{proof}
%This is obtained by substituting $X_k \sim \text{Exp}(\mu)$ in Theorem~\ref{thm:MG11}. The details are omitted due to lack of space. 
%The proof is omitted due to space limitation.
%The proof is given in technical report~\cite{techreport}.
The proof is given in Appendix~\ref{proof:cor:MM11}.
\end{proof}

\begin{comment}
\begin{figure}
\centering
\includegraphics[width = 2.7in]{20180706_VioProbMM11.eps}
\caption{Violation probability vs the arrival rate for an M/M/1/1 system for different age limits $d$ and $\mu = 1$. }
\label{fig:Case2}
\vspace{-.5cm}
\end{figure}
\end{comment}
Using the distribution in Corollary~\ref{cor:MM11}, we compute the expected AoI {\color{black}(cf. Appendix~\ref{EAoI_MM11})} for the M/M/1/1 system to be equal to $\left(\frac{1}{\lambda} + \frac{2}{\mu} - \frac{1}{\mu + \lambda}\right)$, a result reported in~\cite{Costa_2016}.

In the following corollary, we derive the violation probability for the system with zero-wait policy and exponentially distributed service times.
\begin{corollary}
For the system with zero-wait policy and exponentially distributed service times, given $d \geq 0$, the violation probability is given by 
\begin{align}\label{eq:just-in-time:M}
\P(\Delta > d) = (1+\mu d)e^{-\mu d}, \text{ a.s.}
\end{align}
\end{corollary}
\begin{proof}
The result can be obtained from Corollary~\ref{cor:MM11} by utilizing the fact that the statistics of this system will be same as that for $M/M/1/1$ when $\lambda$ approaches infinity.
%refer technical report for complete proof.
%Recall that the statistics of this system will be same as that for $M/M/1/1$ when $\lambda$ approaches infinity. Therefore, we have w.p.1
%\begin{align*}
%&\P(\Delta > d) = \lim_{\lambda \rightarrow \infty}  \frac{\mathbb{E}[g(k)]}{\mathbb{E}[X]+\mathbb{E}[I_k]}\\
%&= \lim_{\lambda \rightarrow \infty}  \mu \mathbb{E}[g(k)] \\
%&=\lim_{\lambda \rightarrow \infty} \mu\left[\frac{\mu^2(e^{-\lambda d}-e^{-\mu d})}{\lambda^3(\mu/\lambda-1)^2} +e^{-\mu d}\left(\frac{1}{\lambda} + \frac{1}{\mu} - \frac{d}{\mu/\lambda - 1}\right)\right]\\
%&=\mu\left[e^{-\mu d}\left(\frac{1}{\mu} + d\right)\right] = (1+\mu d)e^{-\mu d}.
%\end{align*}
\end{proof}
Interestingly, the distribution in~\eqref{eq:just-in-time:M} is gamma distribution with shape parameter $2$ and scale parameter $\frac{1}{\mu}$. Further, the expected AoI in this case is $\frac{2}{\mu}$, a result reported in~\cite{kaul_2012a,Costa_2016}.

\subsection{Upper Bound for the GI/GI/1/1 system}
In this section, we provide an upper bound for the violation probability for the GI/GI/1/1 system, and also analyse its performance. To this end, we first provide an upper bound for $g(k)$ in the following lemma.
\begin{lemma}\label{lem:GammakGG11}
For a GI/GI/1/1 system, $g(k) \leq \Gamma_1(k)$ for all $k$, where 
\begin{align}\label{eq:gkUBGG11}
\Gamma_1(k) = \min\left\{(X_{k-1} + \check{Z}_k + X_k -d)^+,X_k + \check{Z}_k\right\}.
\end{align}
\end{lemma}
\begin{proof}
Recall that $\check{Z}_{k}$ is the inter-arrival time between packet $k$ and its previous arrival. Therefore, we have $I_k \leq \check{Z}_{k}$.  The result follows from using this in~\eqref{eq:gkGG11}.
\end{proof}
{\color{black} \textbf{\textit{Remark 1:}} In an M/G/1/1 system $\mathbb{E}[\Gamma_1(k)] = \mathbb{E}[g(k)]$ as both $I_k$ and $\check{Z}_k$ have the same distribution $\text{Exp}(\lambda)$. Thus, $\mathbb{E}[\Gamma_1(k)]$ is a tight upper bound for $\mathbb{E}[g(k)]$ for the GI/GI/1/1 system.}
% Therefore, choosing $\hat{\nu} = \nu$,  $\Phi_1$ will be an exact solution to the violation probability for this system.

%\begin{align*}
%\hat{\nu} \geq \frac{1}{\mathbb{E}[X] + \mathbb{E}[I_k]}=\nu.
%\end{align*}
%We now define $\Phi_1$ as follows:
%\begin{align}\label{eq:Phi1}
%\Phi_1 \triangleq \hat{\nu} \mathbb{E}[\Gamma(k)], 
%\end{align}
%where $\Gamma(k)$ is given in~\eqref{eq:gkUBGG11}. In the following lemma we establish that $\Phi_1$ is an upper bound for the violation probability.
%It is easy to see that $\Phi_1$ is an upper bound for the violation probability. 
%Let $\hat{\nu}$ denote an upper bound on the expected departure rate $\nu$.
The following theorem presents an upper bound $\Phi_1$ for the violation probability.
\begin{theorem}\label{thm:GG11Bounds}
For a GI/GI/1/1 system, given $d > 0$, assuming that the AoI process is stationary, the violation probability is bounded as follows:
\begin{align*}
\P(\Delta > d) = \nu\mathbb{E}[\gamma^*(k)]  \leq  \Phi_1,
\end{align*}
where $\gamma^*$ is given by Lemma~\ref{lem:gkLB}, and $\Phi_1 = \hat{\nu} \mathbb{E}[\Gamma_1(k)]$, for some $\hat{\nu} \geq \nu$, where $\nu$ is given in~\eqref{eq:nuGG112}.
\end{theorem}
\begin{proof}
The equality follows from the fact that $\gamma^*(k)$ is equal to $g(k)$ given in~\eqref{eq:gkGG11} for the GI/GI/1/1 system. It is easy to see that $\Gamma_1(k)$ are s.i.i.d., and as noted in the proof of Theorem~\ref{thm:GG11}, $\TD(k) - \TD(k-1)$ are i.i.d. Therefore, from Theorem~\ref{thm:evaluategk} we infer that 
\begin{align*}
\lim_{T \rightarrow \infty} \frac{1}{T} \sum_{k = 1}^{K(T)} \Gamma_1(k) &= \nu\mathbb{E}[ \Gamma_1(k)], \text{ a.s.} 
\end{align*}
Using the above equation in Theorem~\ref{thm:vioprobAoIPeakUB}, we obtain $\P(\Delta > d ) \leq  \nu\mathbb{E}[ \Gamma_1(k)]$.
%\begin{align*}
%\lim_{t \rightarrow \infty} \P\{\Delta(t) > d \} \leq  \nu\mathbb{E}[ \Gamma_1(k)]
%\end{align*}
The result follows as $\hat{\nu} \geq \nu$.
\end{proof}
%We note that, $\Phi_1$ is the best achievable upper bound among the set of upper bounds that do not use any of the statics of the idle time. 
We define $\eta$ below that will be used in describing the worst-case performance of $\Phi_1$.
\begin{align}\label{eq:eta}
\eta \triangleq \frac{1}{\lambda} + \frac{1}{\mu} - \frac{1}{\nu}.
\end{align}
In the following corollary we present a worst-case-performance guarantee for $\Phi_1$.
\begin{theorem}\label{thm:worst-caseUBGG11}
For a GI/GI/1/1 system, for a given $\hat{\nu} \ge \nu$, $\Phi_1$ has the following worst-case-performance guarantee .
\begin{align*}
\Phi_1 \leq \frac{\hat{\nu}}{\nu} \cdot \P(\Delta > d) + \hat{\nu}\eta.
\end{align*}
%where $\eta = \frac{1}{\lambda} + \frac{1}{\mu} - \frac{1}{\nu}$.
\end{theorem}
\begin{proof}
Noting that $I_k \leq \check{Z}_k$, we have
\begin{align*}
\Gamma_1(k) &= \min\left\{(X_{k-1} + \check{Z}_k + X_k -d)^+,X_k + \check{Z}_k\right\}\\
&\leq \! \min\left\{(X_{k-1}\! +\! I_k\! +\! X_k -d)^+,X_k\! +\! I_k\right\} + (\check{Z}_k \! -\! I _k)\\
&\leq g(k) + (\check{Z}_k  - I _k).
\end{align*}
Therefore, using Theorem~\ref{thm:GG11Bounds}, we obtain
\begin{align*}
\Phi_1 &\leq \hat{\nu}(\mathbb{E}[g(k)] + \mathbb{E}[\check{Z}_k] - \mathbb{E}[I_k])\\
&= \frac{\hat{\nu}}{\nu}\cdot \P(\Delta(t) > d) + \hat{\nu}\left(\frac{1}{\lambda} + \frac{1}{\mu} - \frac{1}{\nu}\right).
\end{align*}
In the last step above we have used Theorem~\ref{thm:GG11} and~\eqref{eq:nuGG112}.
\end{proof}
From Theorem~\ref{thm:worst-caseUBGG11}, we infer that if $\hat{\nu} = \nu$, i.e., the departure rate is given, then $\Phi_1$ overestimates the violation probability by at most $\eta$. We note that $\frac{1}{\lambda} + \frac{1}{\mu} \geq \frac{1}{\nu}$, and the relation holds with equality for an M/GI/1/1 system. Further, $\nu$ increases sub-linearly with $\lambda$ in a GI/GI/1/1 system, in general. For example, $\nu = \lambda(1-e^{-\mu/\lambda})$ for the D/M/1/1 system (Corollary~\ref{cor:DM11}). Therefore, for a fixed $\mu$, $\eta$ decreases with $\lambda$, in general. In other words, the derived upper bound is tighter at higher utilization. 
Finally, the worst-case guarantee in Theorem~\ref{thm:worst-caseUBGG11} is provided for any $d \geq 0$. Therefore, we expect that $\Phi_1$ may not be tight for larger $d$ values for which the violation probability takes smaller values. 
%We note however that tighter bounds with better worst-case guarantees may be possible for 

We require to compute the value of expected idle time to obtain $\nu$. When $\nu$ is not tractable, we propose to use $\hat{\nu} = \min\{\lambda,\mu\}$, a trivial upper bound on the departure rate.
%\begin{align}\label{eq:hatnu}
%\hat{\nu} = \min\{\lambda,\mu\}.
%\end{align}
We note however that the conclusion about tightness of the upper bound at higher utilization may no longer be valid in this case.

\section{The GI/GI/1/2* System}\label{sec:GG12}
The analysis of a GI/GI/1/2* system follows similar steps to the analysis we have presented for the GI/GI/1/1 system. We first obtain expressions for $I_k$ and $A^\text{peak}(k)$, and use them to obtain $g(k)$.
%, using which we obtain $g(k)$. 

In Figure~\ref{fig:GG12}, we present a possible sequence of arrivals (in blue) and departures (in red) in a GI/GI/1/2* system. Note that there are no arrivals during the service of packet $(k-1)$. This happens only when $\hat{Z}_{k-1} > W_{k-1} + X_{k-1}$ and in this case, $I_k = \hat{Z}_{k-1} - W_{k-1} - X_{k-1}$. If $\hat{Z}_{k-1} \leq W_{k-1} + X_{k-1}$, then $I_k = 0$. Therefore, we have 
\begin{align}\label{eq:IkGG12}
I_k = (\hat{Z}_{k-1} - X_{k-1} - W_{k-1})^+.
\end{align}
Recall that $A^\text{peak}(k) = \TD(k) - \TA(k-1)$. From Figure~\ref{fig:GG12}, it is easy to infer that $A^\text{peak}(k) = X_k + X_{k-1} + I_k + W_k$.  
The following lemma immediately follows from the above analysis and Lemma~\ref{lem:gk}.
\begin{lemma}\label{lem:gkGG12}
Given $d \geq 0$, for any sample path of $\Delta(t)$ in a GI/GI/1/2* system, we have for all $k$,
\begin{align}\label{eq:gkGG12}
\!g(k)\! =\! \min\left\{(X_k\! +\! X_{k-1}\! +\! I_k\! +\! W_{k-1}\! -d)^+,X_k\! +\! I_k\right\}.
\end{align}
\end{lemma}

\begin{figure}[t]
\centering
\includegraphics[width = 3.2in]{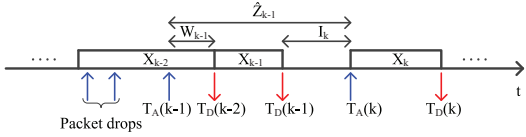}
\vspace{-.5cm}
\caption{An example illustration of arrivals and departures in a GI/GI/1/2* system.}
\label{fig:GG12}
\vspace{-.4cm}
\end{figure}

\begin{comment}
We now provide a general expression for the violation probability in the following theorem. 

\begin{theorem}\label{thm:GG12}
For a GI/GI/1/2* system, assuming AoI process is stationary and ergodic, given $d \geq 0$, $\lambda > 0$, and $0 < \mathbb{E}[X] = \frac{1}{\mu} < \infty$, the violation probability, if exists, is given by $\nu \mathbb{E}[g(k)]$, almost surely, 
%\begin{align*}
%\P(\Delta > d) = \nu \mathbb{E}[g(k)] \, \text{ a.s.},
%\end{align*}
where $g(k)$ is given by~\eqref{eq:gkGG12}, $\nu$ is given by~\eqref{eq:nuGG112}, and $I_k$ is given by~\eqref{eq:IkGG12}.
\end{theorem}
\begin{proof}
The proof follows similar steps to the proof of Theorem~\ref{thm:GG11} and is omitted.
\end{proof}
\end{comment}

Unlike the case of the GI/GI/1/1 system, for the GI/GI/1/2* system it is hard to derive a closed-form expression for the violation probability in terms of $X_k$, $X_{k-1}$, $I_k$ and $W_{k-1}$, because $g(k)$, given in~\eqref{eq:gkGG12}, does not satisfy the s.i.i.d. property. Further, computing the violation probability requires the distributions of both $I_k$ and $W_{k-1}$. While these quantities can be computed for exponential service or exponential inter-arrival times (cf \cite{Yoshiaki2018}), they become intractable for general inter-arrival and service-time distributions.  To this end we present upper bounds in the next section.
\subsection{Upper Bound for the GI/GI/1/2* system}
In this subsection we propose an upper bound for the violation probability and analyse its worst-case performance.
\begin{lemma}\label{lem:GammakGG12}
For a GI/GI/1/2* system, given $d \geq 0$, $g(k) \leq \Gamma_2(k)$ for all $k$, where 
\begin{align*}
\Gamma_2(k)\! = \! \min\!\{\!(X_k \! + \! X_{k-1} \! + \! \hat{Z}_{k-1}\!  - d)^+\!,\!X_k \! + \! (\hat{Z}_{k-1} \! - \! X_{k-1})^+\!\}.
\end{align*}
\end{lemma}
\begin{proof}
Noting the expression for $g(k)$ given in~\eqref{eq:gkGG12}, it is sufficient to show that $I_k+W_{k-1} \leq \hat{Z}_{k-1}$, and $I_k \leq (\hat{Z}_{k-1} - X_{k-1})^+$. The latter inequality follows from~\eqref{eq:IkGG12}. The former inequality is obviously true if there are no arrivals during the service of packet $(k-1)$; see Figure~\ref{fig:GG12}. If there is an arrival during the service of packet $(k-1)$, then $I_k = 0$. In this case $I_k+W_{k-1} = W_{k-1} < \hat{Z}_{k-1}$, since by definition there should be no arrival after packet $(k-1)$ arrived and before its service started.
\end{proof}

In the following theorem we present an upper bound $\Phi_2$ for the violation probability.
\begin{theorem}\label{thm:GG12Bounds}
For a GI/GI/1/2* system, assuming that the AoI process is stationary, the violation probability is bounded by,
\begin{align*}
\nu\mathbb{E}[\gamma^*(k)] \leq \P(\Delta > d) \leq  \Phi_2,
\end{align*}
where $\Phi_2 = \hat{\nu} \mathbb{E}[\Gamma_2(k)]$, for some $\hat{\nu} \geq \nu$.
\end{theorem}
\begin{proof}
The proof follows similar steps to the proof of Theorem~\ref{thm:GG11Bounds} and is omitted.
\end{proof}

A worst-case-performance guarantee for $\Phi_2$ is presented in the following theorem.
\begin{theorem}\label{thm:worst-caseUBGG12}
For the GI/GI/1/2* system, for a given $\hat{\nu} \ge \nu$, $\Phi_2$ has the following worst-case-performance guarantee.
\begin{align*}
\Phi_2 \leq \frac{\hat{\nu}}{\nu}\cdot \P(\Delta > d) + \hat{\nu}\eta.
\end{align*}
\end{theorem}
\begin{proof}
It is easy to show that $\Gamma_2(k) \leq \gamma^*(k) + \hat{Z}_{k-1} - I_k$. The rest of the proof follows similar steps as in the proof of Theorem~\ref{thm:worst-caseUBGG11} and is omitted.
%Therefore,
%\begin{align*}
%\Phi_2 = \hat{\nu} \mathbb{E}[\Gamma_1(k)] \leq \mathbb{E}[\gamma^*(k)] + \mathbb{E}[\hat{Z}_{k-1}].
%\end{align*}
%The result follows by using the lower bound from Theorem~\ref{thm:GG12Bounds}.
\end{proof}
Thus, $\Phi_2$ also overestimates the violation probability by at most $\eta$, if $\nu$ is given. %Using the same arguments as before, we claim that, given $\nu$ and for a fixed service rate, the overestimation value decreases with $\lambda$, and hence the upper bound is tighter at higher utilization. 
Therefore, given $\nu$ and for a fixed average service, $\Phi_2$ is tighter at higher utilization.
Since it is hard to compute $\nu$, in general, in the numerical section we compute $\Phi_2$ using $\hat{\nu} = \min\{\lambda,\mu\}$. 
%provide the following upper bound.
%\begin{align}\label{eq:nuGG12}
%\hat{\nu} = \min\left \{\lambda,\frac{\mu}{1+\mu \mathbb{E}[(\check{Z}_k - X_{k-1} - X_{k-2})^+]} \right\}.
%\end{align}
%The above upper bound is obtained by using the fact that $W_{k-1} \leq X_{k-2}$ in~\eqref{eq:IkGG12}.

\textit{\textbf{Remark 2:}} For both GI/GI/1/1 and GI/GI/1/2* systems $\nu = 1/(\mathbb{E}[X_k] + \mathbb{E}[I_k])$, and $g(k)$ for GI/GI/1/1 given by~\eqref{eq:gkGG11} seems to be closely related to $g(k)$ for GI/GI/1/2* given by~\eqref{eq:gkGG12}. 
Also, one can expect that the idle time in GI/GI/1/2* will be lower compared to that of GI/GI/1/1.  
However, for a given $d$, a comparison between the violation probabilities in these systems is non-trivial because of the waiting time in GI/GI/1/2* and higher idle time in GI/GI/1/1. 

\textit{\textbf{Remark 3:}} When the input rate approaches infinity, the inter-arrival time, waiting time, and idle time approach zero. Therefore, the upper bounds $\Phi_1$, $\Phi_2$, and the respective violation probabilities in GI/GI/1/1 and GI/GI/1/2*, all converge to the violation probability in the system using zero-wait policy. Thus, both $\Phi_1$ and $\Phi_2$ are asymptotically tight.

\section{Numerical Results}\label{sec:numerical}
In this section we validate the proposed upper bounds against the violation probability obtained through simulation for selected service-time and inter-arrival-time distributions. 
%The numerical results presented here are not comprehensive, and only serve as an initial validation of the upper bounds. 
For all simulations we set $\mu = 1$ and thus the utilization increases with $\lambda$. We use $\lambda = .45$ and $d = 5$ as default values.
%The computation of $\mathbb{E}[\Gamma_1(k)]$ and $\mathbb{E}[\Gamma_2(k)]$ for the systems we consider below are presented in the technical report.

We first study the performance of $\Phi_1$ in comparison with overestimation factor $\eta$, when $\nu$ is given. To this end we consider the D/M/1/1 system and compute $\Phi_1$ by setting $\hat{\nu} = \nu = \lambda(1-e^{-\mu/\lambda})$. In Figure~\ref{fig:nonPreemptiveDM11_varR}, we plot $\Phi_1$ against the exact value for the violation probability given in Corollary~\ref{cor:DM11}. Observe that the gap between $\Phi_1$ and violation probability reduces as the arrival rate increases confirming our initial conclusion that the bound is tighter at higher utilization. Furthermore, $\Phi_1$ approaches the simulated violation probability asymptotically. 
%This is expected as in this case $\nu$ will be equal to $\mu$ and the overestimation factor reduces to zero. 
For $d = 5$ and $\lambda = 0.4$,  we compute $\eta$ to be $0.28$, while the actual gap is $0.08$. For the same setting, but for $d = 10$, $\eta$ remains the same while the actual gap is $0.0012$. 
%For $d = 10$, the overestimation factors remain the same, but the violation probabilities are lower by two orders of magnitude. 
This suggests that the proposed upper bound is much lower than the worst-case-performance guarantee.   
\begin{figure}[t]
\centering
\includegraphics[width = 2.5in]{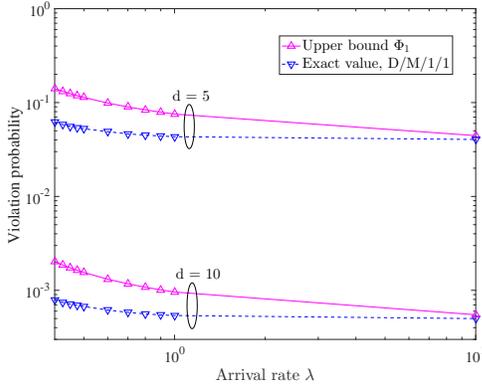}
\caption{Performance of $\Phi_1$ with varying $\lambda$, when $\nu$ is given, and $\mu = 1$}
\vspace{-.3cm}
\label{fig:nonPreemptiveDM11_varR}
\end{figure}

\begin{comment}
\begin{figure}[t]
\centering
\includegraphics[width = 2.5in]{20180720_nonPreemptiveDM11_vard.eps}
\caption{Performance of the upper bounds for varying age limit for D/M/1/1 system with $\lambda = .45$ $\mu = 1$.}
\label{fig:nonPreemptiveDM11_vard}
\end{figure}
\end{comment}

\begin{figure}[t]
\centering
\includegraphics[width = 2.5in]{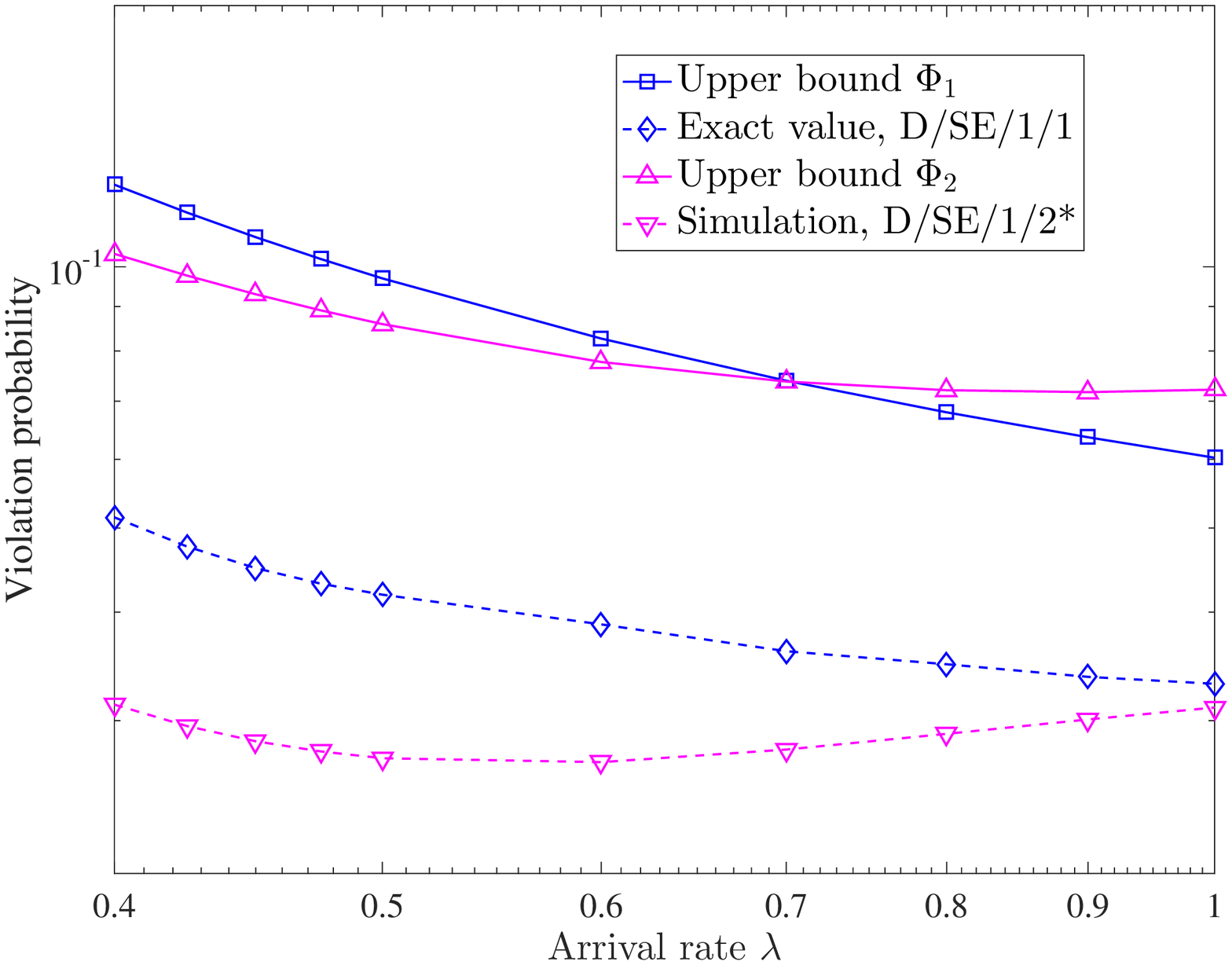}
\caption{Performance of upper bounds with varying $\lambda$, $d = 5$, $\mu = 1$, and shift equal to $0.11$ .}
\vspace{-.4cm}
\label{fig:nonPreemptiveDShiftExp1_varR}
\end{figure}

\begin{figure}[t]
\centering
\includegraphics[width = 2.5in]{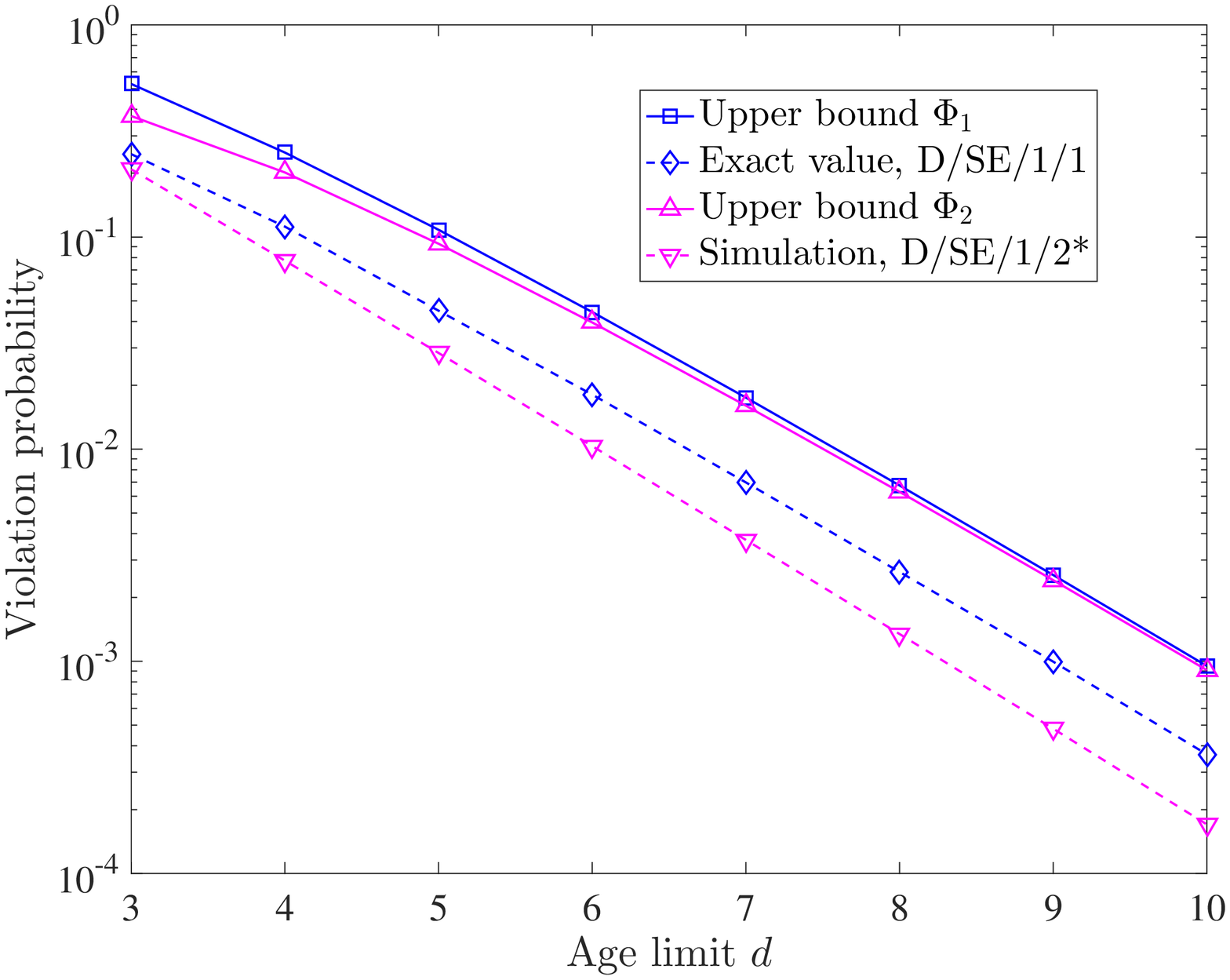}
\caption{Performance of upper bounds with varying $d$, $\lambda = .45$, $\mu = 1$, and shift equal to $0.11$.}
\vspace{-.3cm}
\label{fig:nonPreemptiveDShiftExp1_vard}
\end{figure}

\begin{figure}[t]
\centering
\includegraphics[width = 2.5in]{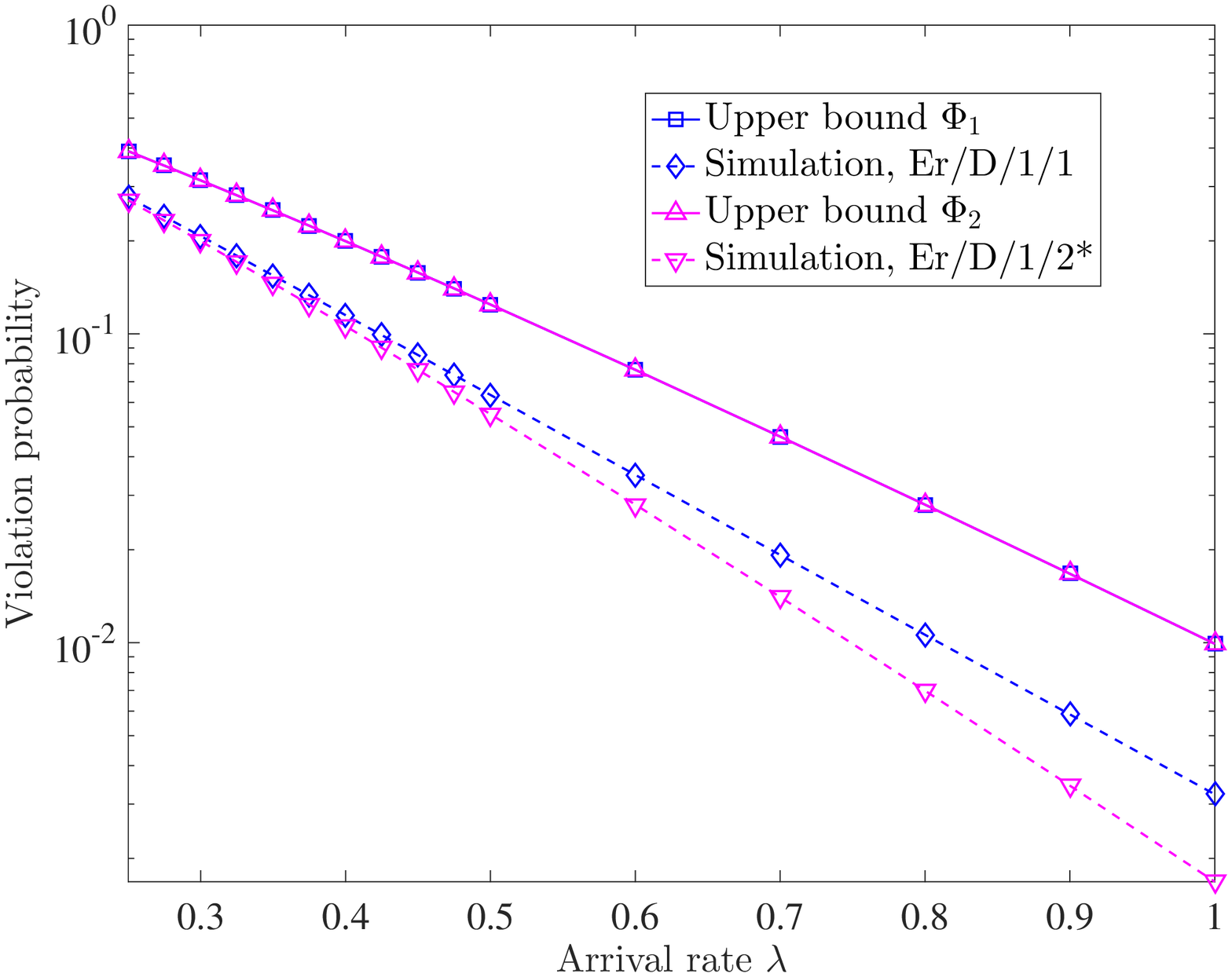}
\caption{Performance of  upper bounds with varying $\lambda$, Erlang shape parameter equal to $2$, $d = 5$, and $\mu = 1$.}
\vspace{-.4cm}
\label{fig:nonPreemptiveED1_varR}
\end{figure}

\begin{figure}[t]
\centering
\includegraphics[width = 2.5in]{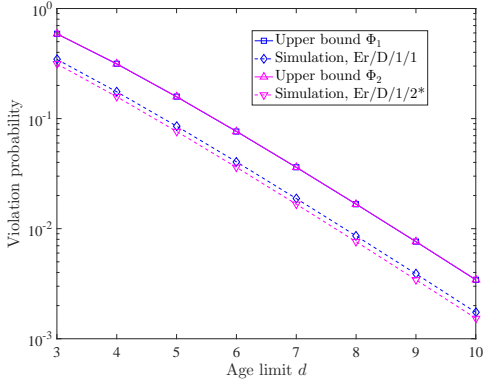}
\caption{Performance of upper bounds with varying $d$, $\lambda = .45$, Erlang shape parameter equal to $2$, and $\mu = 1$.}
\vspace{-.3cm}
\label{fig:nonPreemptiveED1_vard}
\end{figure}

Next, we consider two example systems where exact expressions for the distribution of AoI are hard to compute. 
%In one system we consider deterministic arrivals and in the other system we consider deterministic service. 
For both systems, we use $\hat{\nu} = \min(\lambda,\mu)$ to compute $\Phi_1$ and $\Phi_2$.
In the first example system, we choose deterministic arrivals and Shifted-Exponential (SE) service times, i.e., D/SE/1/1 and D/SE/1/2*.  We set values of $d$ and $\lambda$ such that $d \geq \frac{1}{\lambda}$, $\mu = 1$ and shift parameter equal to $0.11$. 
In Figures~\ref{fig:nonPreemptiveDShiftExp1_varR} and~\ref{fig:nonPreemptiveDShiftExp1_vard}, we study the performance of the upper bounds, presented in Theorems~\ref{thm:GG11Bounds} and~\ref{thm:GG12Bounds}, for varying arrival rate $\lambda$ and varying age limit $d$, respectively. %for deterministic arrivals and shifted exponential service with mean $1$ and shift equal to $.11$. 
From Figure~\ref{fig:nonPreemptiveDShiftExp1_varR}, we again observe that the upper bounds are tighter at higher utilization.  For $\lambda > 1$ both upper bounds and the violation probabilities  converge to $0.029$. Interestingly, in contrast to D/SE/1/1 where the violation probability decreases with $\lambda$, D/SE/1/2* has minimum violation probability of $0.026$ at around $\lambda = 0.6$.
From Figure~\ref{fig:nonPreemptiveDShiftExp1_vard}, we observe that both bounds are tighter at smaller $d$ values. While the decay rate of $\Phi_1$ matches with that of the simulated violation probability, $\Phi_2$ becomes loose as $d$ increases. We conjecture that this is due to the inequality $I_k + W_{k-1} \leq \hat{Z}_{k-1}$ that we use to obtain this bound. 

In Figures~\ref{fig:nonPreemptiveED1_varR} and~\ref{fig:nonPreemptiveED1_vard}, we present a comparison for deterministic service and Erlang distributed inter-arrival times, i.e., Er/D/1/1 and Er/D/1/2*. We first note that for the parameter values chosen, $\Phi_1$ and $\Phi_2$ are equal in this case.  From Figure~\ref{fig:nonPreemptiveED1_varR}, we observe that the bounds are not tight at larger arrival rate. This can be attributed to the use of $\hat{\nu} = \min(\lambda,\mu)$. From Figure~\ref{fig:nonPreemptiveED1_vard}, we observe that the decay rate of the bounds matches the decay rate of the violation probabilities.
%However, we note that, for $\lambda > 1$ both upper bounds and the violation probabilities quickly converge with the increase in $\lambda$. 
Finally, it is worth noting that, the violation probability in -/-/1/2* is lower than that in -/-/1/1 for the above example systems. 
%Similar observation was made in~\cite{Costa_2016} when they studied the expected AoI for these systems, where they report that the expected AoI in M/M/1/2* system is lower than that of M/M/1/1 system.
%Also, from both Figures~\ref{fig:nonPreemptiveED1_varR} and~\ref{fig:nonPreemptiveED1_vard}, we infer that $\Phi_2$ is tighter in this system compared to $\Phi_1$. 

In conclusion, for the considered systems, the upper bounds are well within an order of magnitude from the violation probability. For most cases the decay rate of the proposed bounds follow the decay rate of the simulated violation probability as $d$ increases. Also, the performance of these upper bounds can be improved further by finding non-trivial upper bounds for $\nu$.  Thus, we believe that the proposed upper bounds can be useful as first-hand metrics for measuring freshness of status updates in these systems.
\vspace{-.2cm}

\section{Conclusion and Future Work}\label{sec:conclusion}
In this work we have studied the distribution of AoI for GI/GI/1/1 and GI/GI/1/2* systems. Toward this end, we first established a fundamental result that is valid for any single-source-single-server queuing system with general service-time and general inter-arrival time distributions. Using this result we have derived exact expressions for the distribution of AoI for D/GI/1/1 and M/GI/1/1 systems. Further, we have proposed and analysed upper bounds for the violation probability for GI/GI/1/1 and GI/GI/1/2* when the distribution of AoI is intractable. An interesting feature of the proposed upper bounds is that, if the departure rate $\nu$ is given, then they overestimate the violation probability by at most $\eta$ which decreases with $\lambda$. Initial numerical results show that the bounds are well within an order of magnitude, and their decay rates closely match, in comparison with the simulated violation probability.

%As noted before, the numerical results presented in Section~\ref{sec:numerical} are not comprehensive and are meant as a proof of concept. 
We note that the close-form expressions and the bounds can be used to study the AoI performance of  a wide range of systems which we aim to do in our future work. We also aim to improve the upper bounds by finding tighter bounds for $g(k)$ and $\nu$. Furthermore, we are interested in studying the systems under pre-emptive scheduling.

%\input{nonpreemptiveGG1}

%\input{preemptive}

%\bibliographystyle{IEEEtran}
%\bibliography{main}
% Generated by IEEEtran.bst, version: 1.14 (2015/08/26)

{\color{black} \appendix

\subsection{Proof of Lemma~\ref{lem:existenceDG11}}\label{proof:lem:existenceDG11}
We prove that $\P(\Delta > d)$ does not exist when $d < \frac{1}{\lambda}$. Consider the event $\{\Delta(t) > d\}$ at time $t$. If $d < \frac{1}{\lambda}$, there will be time instances, say $\hat{t}$, for which there is no arrival in the interval $[\hat{t}-d,\hat{t})$. This implies that at $\hat{t}$ the receiver/destination cannot have a packet with arrival time greater than $\hat{t} - d$. Therefore, the event $\{\Delta(\hat{t}) > d\}$ is true for all such $\hat{t}$. Let $\bar{t}$ denote any time instance $t \neq \hat{t}$, i.e., at $\bar{t}$ there exists an arrival in the interval $[\bar{t}-d,\bar{t})$. Since $d < \frac{1}{\lambda}$, there can be only one arrival in this interval. Therefore, for this case the event $\{\Delta(\hat{t}) > d\}$ is true if either the server is busy, in which case the packet is dropped, or the departure time of this packet exceeds $\bar{t}$.

From the above analysis, we conclude that $\P(\Delta(t) > d)$ depends on the value of $t$. Specifically, we infer that $\limsup_{t \rightarrow \infty} \P(\Delta(t) > d) = 1$, because the event $\{\Delta(\hat{t}) > d\}$ is true for all $\hat{t}$, which occur infinitely often as $t$ goes to infinity. Similarly, we infer that $\liminf_{t \rightarrow \infty} \P(\Delta(t) > d) < 1$, because the time instances $\bar{t}$ also occur infinitely often and at these time instances the occurrence of the event $\{\Delta(\bar{t}) > d\}$ is uncertain. Since the limit supremum and limit infimum are not equal  $\P(\Delta > d) = \lim_{t \rightarrow \infty} \P(\Delta(t) > d)$ does not exist for $d < \frac{1}{\lambda}$.

\subsection{Proof of Lemma~\ref{lem:g(k)Distr}}\label{proof:lem:g(k)Distr}
From~\eqref{eq:gkGG11}, we have
{\allowdisplaybreaks \begin{align*}
&\P\{g(k) > y\} \\
&= \P\{\min\{(X_{k-1}+X_{k}+I_k-d)^+,X_k+I_k\} > y\}\nonumber\\
&=  \P\{\max\{0,X_{k-1}+X_{k}+I_k-d\} > y,X_k+I_k > y\}\nonumber\\
&=  \P\{(y < 0, X_k+I_k > y) \nonumber\\ & \quad\quad\quad \cup(X_{k-1}+X_{k}+I_k-d > y,X_k+I_k > y)\}\nonumber \\
&=  \P\{X_{k-1}+X_{k}+I_k-d > y,X_k+I_k > y\} \\
&= \int_{0}^{\infty}\P\{X_{k}+I_k > y + d -x,X_k + I_k > y\}f_X(x)dx \\
&= \int_{0}^{d} \P\{X_k + I_k > y-x+d\}f_X(x)dx \\
&\quad\quad\quad + \int_{d}^{\infty}\P\{X_k+I_k>y\}f_X(x) dx.
\end{align*}}

\subsection{Proof of Corollary~\ref{cor:MM11}}\label{proof:cor:MM11}
Since $X_k \sim Exp(\mu)$ and $I_k \sim Exp(\lambda)$, we have
\begin{eqnarray}\label{eq:Xk+Ik}
\P(X_k + I_k > y) \! = \! \left\{\begin{array}{lc}
	 \!\frac{\mu e^{-\lambda(y-x+d)} - \lambda e^{-\mu(y-x+d)}}{\mu - \lambda} & \lambda \neq \mu,\\
   \!(1+\mu y)e^{-\mu y} & \lambda = \mu.
  \end{array}\right.
\end{eqnarray}
In the following we compute the distribution of $g(k)$by substituting~\eqref{eq:Xk+Ik} in $\P\{g(k) > y\}$ given in Lemma~\ref{lem:g(k)Distr}.

\textbf{Case 1:} $\mu \neq \lambda$. For this case, we have
{\allowdisplaybreaks \begin{align*}
&\P(g(k) > y) \\
&= \int_0^d \frac{\mu e^{-\lambda(y-x+d)} - \lambda e^{-\mu(y-x+d)}}{\mu - \lambda} f_X(x) dx \\
&\quad\quad\quad\quad\quad\quad + \frac{\mu e^{-\lambda y} - \lambda e^{-\mu y}}{\mu - \lambda}\int_d^{\infty}f_X(x)dx \\
&= \frac{\mu e^{-\mu(y+d)}}{\mu - \lambda}\left[\frac{\mu (e^{d(\lambda - \mu)}-1)}{\lambda - \mu} - \lambda d\right]\\
&\quad\quad\quad\quad\quad\quad + \frac{(\mu e^{-\lambda y} - \lambda e^{-\mu y})e^{-\mu d}}{\mu - \lambda}.
\end{align*}}
Integrating the above expression over $y$, we obtain the desired result.

\textbf{Case 2:} $\mu = \lambda$. For this case, we have
{\allowdisplaybreaks \begin{align*}
&\P(g(k) > y) \\
&= \int_0^d (1+\mu(y-x+d))e^{-\mu(y-x+d)}f_X(x)dx\\
&\quad\quad\quad\quad\quad\quad + \int_d^{\infty}(1+\mu y)e^{-\mu y} f_X(x) dx \\
&= \mu e^{-\mu(y+d)} \int_0^d (1+\mu(y-x+d)) dx + (1+\mu y)e^{-\mu (y+d)}\\
&= \mu e^{-\mu(y+d)} \left[(1+\mu(y+d))d - \frac{\mu d^2}{2}\right]+ (1+\mu y)e^{-\mu (y+d)}\\
&= \mu d e^{-\mu(y+d)} \left[1+\mu\left(y+\frac{d}{2}\right)\right]+ (1+\mu y)e^{-\mu (y+d)}\\
&= (\mu d+1)(1+\mu y)e^{-\mu (y+d)} + \frac{\mu^2 d^2}{2}e^{-\mu (y+d)}.
\end{align*}}
Therefore, integrating the above expression over $y$, we obtain
\begin{align*}
\mathbb{E}[g(x)] &= e^{-\mu d}\left[(\mu d + 1) \frac{2}{\mu} + \frac{\mu d^2}{2}\right]\\
&= \frac{\mu e^{-\mu d}}{2}\left(d+\frac{2}{\mu}\right)^2
\end{align*}

\subsection{Computation of $\mathbb{E}[\Delta]$ for M/M/1/1}\label{EAoI_MM11}
\textbf{Case 1:} $\lambda \neq \mu$. Recall that $\frac{1}{\nu} = \frac{1}{\lambda} + \frac{1}{\mu}$. Using $\mathbb{E}[g(k)]$ from Corollary~\ref{cor:MM11}, we obtain
\begin{align*}
&\mathbb{E}[\Delta]\! =\! \nu \! \!\int_0^\infty \!\!\left[\frac{\mu^2(e^{-\lambda y}-e^{-\mu y})}{\lambda(\mu - \lambda)^2}\! + \! e^{-\mu y}\left(\frac{1}{\lambda}\! +\!\frac{1}{\mu}\! -\! \frac{\lambda y}{\mu - \lambda}\right)\right]\!dy \\
&= \frac{\mu^2 \nu}{\lambda(\mu - \lambda)^2} \left(\frac{1}{\lambda} - \frac{1}{\mu}\right) + \left(\frac{1}{\lambda} + \frac{1}{\mu}\right)\frac{\nu}{\mu}\\ 
& \quad \quad - \frac{\lambda \nu}{\mu - \lambda}\left[\frac{-y e^{-\mu y}}{\mu}-\frac{e^{-\mu y}}{\mu^2}\bigg{|}_0^\infty\right]\\
&= \frac{\mu^2}{\lambda(\mu^2 - \lambda^2)} + \frac{1}{\mu} - \frac{\lambda^2}{\mu(\mu^2 - \lambda^2)}\\
&= \frac{\mu^3 - \lambda^3}{\lambda \mu (\mu^2 - \lambda^2)} + \frac{1}{\mu} = \frac{\mu^2 + \lambda^2 + \lambda \mu}{\lambda \mu (\mu + \lambda)} + \frac{1}{\mu}\\
&= \left(\frac{1}{\lambda} + \frac{2}{\mu} - \frac{1}{\mu + \lambda}\right).
\end{align*}

\textbf{Case 2:} $\lambda = \mu$. For this case $\nu = \frac{\mu}{2}$. Using $\mathbb{E}[g(k)]$ from Corollary~\ref{cor:MM11}, we obtain
\begin{align*}
&\mathbb{E}[\Delta] = \nu \int_0^\infty \frac{\mu e^{-\mu y}}{2}\left(y+\frac{2}{\mu}\right)^2 dy\\
&= \frac{\mu^2}{2} \int_0^\infty e^{-\mu y}\left(y+\frac{2}{\mu}\right)^2 dy \\
&= \frac{\mu^2e^2}{2} \int_{\frac{2}{\mu}}^\infty e^{-\mu z}z^2 dz \\
&= \frac{\mu^2e^2}{2} \left[\frac{e^{-\mu z}}{\mu} (-z^2-2(z/\mu+1/\mu^2))\bigg{|}_{z=\frac{2}{\mu}}^{\infty}\right] 
&= \frac{5}{2 \mu}.
\end{align*}
It is easy to verify that $\left(\frac{1}{\lambda} + \frac{2}{\mu} - \frac{1}{\mu + \lambda}\right)$ is equal to $\frac{5}{2 \mu}$ for $\lambda = \mu$.

\subsection{Proof of Corollary~\ref{cor:DM11}}\label{proof:cor:DM11}
%Evaluating A
\begin{figure*}[!t]
\normalsize
\setcounter{equation}{20}
% Store the current equation number.
%\setcounter{mytempeqncnt}{\value{equation}}
% Set the equation number to one less than the one
% desired for the first equation here.
% The value here will have to changed if equations
% are added or removed prior to the place these
% equations are referenced in the main text.
%\setcounter{equation}{16}
\begin{align*}
A &= \int_{0}^{d}\P\left\{X_{k} > y + d  - \frac{\lceil \lambda x \rceil}{\lambda}\right\}f_X(x)dx \\
&= \int_{0}^{\frac{\lfloor \lambda d\rfloor}{\lambda}}\P\left\{X_{k} > y + d  - \frac{\lceil \lambda x \rceil}{\lambda}\right\}f_X(x)dx + \int_{\frac{\lfloor \lambda d\rfloor}{\lambda}}^{d}\P\left\{X_{k} > y + d  - \ceild\right\}f_X(x)dx \\
&= \sum_{m=1}^{\lfloor \lambda d\rfloor} \int_{\frac{m-1}{\lambda}}^{\frac{m}{\lambda}}\P\left\{X_{k} > y + d  - \frac{m}{\lambda}\right\}f_X(x)dx + \P\left\{X_{k} > y + d  - \ceild\right\} (e^{-\mu \frac{\lfloor \lambda d\rfloor}{\lambda}} - e^{-\mu d}) \\
&= \sum_{m=1}^{\lfloor \lambda d\rfloor} e^{-\frac{\mu  m}{\lambda}}\int_{\frac{m-1}{\lambda}}^{\frac{m}{\lambda}}f_X(x)dx + \P\left\{X_{k} > y + d  - \ceild\right\} (e^{-\mu \frac{\lfloor \lambda d\rfloor}{\lambda}} - e^{-\mu d}) \\
&= e^{-\mu(y+d)}\sum_{m=1}^{\lfloor \lambda d\rfloor} e^{\frac{\mu  m}{\lambda}} (e^{\frac{-\mu(m-1)}{\lambda}} -e^{-\frac{\mu m}{\lambda}}) + \P\left\{X_{k} > y + d  - \ceild\right\} (e^{-\mu \frac{\lfloor \lambda d\rfloor}{\lambda}} - e^{-\mu d}) \\
&= e^{-\mu(y+d)}(e^{\frac{\mu}{\lambda}}-1)\lfloor \lambda d\rfloor + \P\left\{X_{k} > y + d  - \ceild\right\} (e^{-\mu \frac{\lfloor \lambda d\rfloor}{\lambda}} - e^{-\mu d}).
\end{align*}
% IEEE uses as a separator
\hrulefill
% The spacer can be tweaked to stop underfull vboxes.
\vspace*{4pt}
\end{figure*}

%Integrating A
\begin{figure*}[!t]
\normalsize
\begin{align*}
\int_{0}^{\infty} A dy &= e^{-\mu d}(e^{\frac{\mu}{d}}-1)\lfloor \lambda d\rfloor \int_{0}^{\infty}e^{-\mu y} dy 
+ (e^{-\mu \frac{\lfloor \lambda d\rfloor}{\lambda}} - e^{-\mu d})\left[\int_{0}^{\ceild-d}dy  
 + e^{-\mu(d-\ceild)}\int_{\ceild-d}^{\infty}e^{-\mu y}dy \right] \\
&= \frac{1}{\mu}e^{-\mu d}(e^{\frac{\mu}{d}}-1)\lfloor \lambda d\rfloor + (e^{-\mu \frac{\lfloor \lambda d\rfloor}{\lambda}} - e^{-\mu d})\left[\ceild-d + \frac{1}{\mu}\right]
\end{align*}
% IEEE uses as a separator
\hrulefill
% The spacer can be tweaked to stop underfull vboxes.
\vspace*{4pt}
\end{figure*}

%Evaluating B
\begin{figure*}[!t]
\normalsize
\begin{align*}
B &= \int_{d}^{\infty}\P\left\{X_k > y + x  - \frac{\lceil \lambda x \rceil}{\lambda} \right\}f_X(x)dx \\
&= \int_{d}^{\ceild}\P\left\{X_k > y + x  - \ceild \right\}f_X(x)dx 
+ \sum_{m=1}^{\infty} \int_{\beta_{m-1}}^{\beta_m}\P\left\{X_k > y + x  - \beta_m \right\}f_X(x)dx\\
&\quad \quad \quad \quad(\text{ where  } \beta_m = \frac{\lceil \lambda d\rceil + m}{\lambda})\\
&= \int_{d}^{\max\{d,\ceild - y\}}f_X(x)dx 
+ \int_{\max\{d,\ceild - y\}}^{\ceild}e^{-\mu(y-\ceild)}\mu e^{-2 \mu x}dx \\
&\quad +  \sum_{m=1}^{\infty} \left[\int_{\beta_{m-1}}^{\max\{\beta_{m-1},\beta_m-y\}}f_x(x) dx + \int_{\max\{\beta_{m-1},\beta_m-y\}}^{\beta_m} e^{-\mu(y-\beta_m)}\mu e^{-2\mu x}dx\right]\\
&= e^{-\mu d} - e^{-\mu \max\{d,\ceild - y\}} 
+ \frac{1}{2}e^{-\mu(y-\ceild)}\left(e^{-2\mu \max\{d,\ceild-y\}}-e^{-2\mu  \ceild}\right )\\
&\quad + \sum_{m=1}^{\infty} \left[e^{-\mu \beta_{m-1}} -e^{-\mu \max\{\beta_{m-1},\beta_m-y\}} + \frac{1}{2} e^{-\mu(y-\beta_m)} (e^{-2\mu \max\{\beta_{m-1},\beta_m-y\}} - e^{-2\mu \beta_m})  \right]
\end{align*}
% IEEE uses as a separator
\hrulefill
% The spacer can be tweaked to stop underfull vboxes.
\vspace*{4pt}
\end{figure*}

%Integral B
\begin{figure*}[!t]
\normalsize
\begin{align*}
\int_{0}^{\infty} B dy &= \int_{0}^{\ceild - d}(e^{-\mu d} - e^{-\mu (\ceild-y)}) 
+ \int_{\ceild - d}^{\infty}\, 0 \, dy \\
&\quad +\frac{e^{\mu \ceild}}{2}\left[\int_{0}^{\ceild - d} e^{-\mu y}e^{-2\mu(\ceild - y)}dy + \int_{\ceild - d}^{\infty} e^{-\mu y}e^{-2\mu d} dy - \int_{0}^{\infty} e^{-\mu y} e^{-2\mu \ceild} dy \right]\\
&\quad + \sum_{m=1}^{\infty} \left[\int_{0}^{\frac{1}{\lambda}}(e^{-\mu \beta_{m-1}} - e^{-\mu (\beta_m-y)})dy + \int_{\frac{1}{\lambda}}^{\infty} \, 0 \, dy \right] \\
&\quad + \sum_{m=1}^{\infty}\frac{e^{\mu \beta_m}}{2}\left[\int_{0}^{\frac{1}{\lambda}}e^{-\mu y}e^{-2 \mu (\beta_m - y)}dy + \int_{\frac{1}{\lambda}}^{\infty} e^{-\mu y}e^{-2\mu \beta_{m-1}}dy - \int_{0}^{\infty}e^{-\mu y}e^{-2 \mu \beta_m}dy\right]\\
&= e^{-\mu d}\left(\ceild - d \right) - \frac{e^{-\mu \ceild}}{\mu} (e^{\mu (\ceild - d)}-1) + \frac{e^{-\mu d} - e^{-\mu \ceild}}{\mu}\\
&\quad + \sum_{m=1}^{\infty} \left[\frac{e^{-\mu \beta_{m-1}}}{\lambda} - \frac{e^{-\mu \beta_m}}{\mu}(e^{\frac{\mu}{\lambda}}-1)\right] + \sum_{m=1}^{\infty}\frac{e^{-\mu \beta_m}}{2}\left[\int_{0}^{\frac{1}{\lambda}}e^{\mu y} dy + e^{\frac{2\mu}{\lambda}}\int_{\frac{1}{\lambda}}^{\infty}e^{-\mu y}dy - \frac{1}{\mu}\right]\\
&= e^{-\mu d}\left(\ceild - d \right) + \sum_{m=1}^{\infty}\left[\frac{e^{-\mu \beta_{m-1}}}{\lambda} - \frac{e^{-\mu \beta_m}}{\mu}(e^{\frac{\mu}{\lambda}}-1)\right] + \sum_{m=1}^{\infty}\frac{e^{-\mu \beta_m}}{\mu}\left(e^{\frac{\mu}{\lambda}}-1\right)\\
&=  e^{-\mu d}\left(\ceild - d \right) + \frac{e^{-\mu \ceild}}{\lambda(1-e^{-\frac{\mu}{\lambda}})}
\end{align*}
% IEEE uses as a separator
\hrulefill
% The spacer can be tweaked to stop underfull vboxes.
\vspace*{4pt}
\end{figure*}

\begin{figure*}[!t]
\normalsize
\begin{align*}
\mathbb{E}[g(k)] &= \int_0^{\infty} P(g(k) > y)dy = \int_0^{\infty} A dy + \int_0^{\infty} Bdy \\
&= \frac{1}{\mu}e^{-\mu d}(e^{\frac{\mu}{d}}-1)\lfloor \lambda d\rfloor + (e^{-\mu \frac{\lfloor \lambda d\rfloor}{\lambda}} - e^{-\mu d})\left[\ceild-d + \frac{1}{\mu}\right] + e^{-\mu d}\left(\ceild - d \right) + \frac{e^{-\mu \ceild}}{\lambda(1-e^{-\frac{\mu}{\lambda}})}\\
&= \frac{e^{-\mu \ceild}}{\lambda(1-e^{-\frac{\mu}{\lambda}})} + e^{-\mu \frac{\lfloor \lambda d\rfloor}{\lambda}}\left[\ceild-d + \frac{1}{\mu}\right] + \frac{e^{-\mu d}}{\mu}\left((e^{\frac{\mu}{\lambda}}-1)\lfloor \lambda d\rfloor - 1\right)
\end{align*}
% IEEE uses as a separator
\hrulefill
% The spacer can be tweaked to stop underfull vboxes.
\vspace*{4pt}
\end{figure*}

In the following we first derive $\mathbb{E}[\lceil \lambda X \rceil]$. 
\begin{align*}
\mathbb{E}[\lceil \lambda X \rceil] &= \int_{0}^{\infty} \lceil \lambda x \rceil \mu e^{-\mu x}dx\\
&= \sum_{m=1}^{\infty}m \int_{\frac{m-1}{\lambda}}^{\frac{m}{\lambda}}\mu e^{-\mu x}dx \\
&= (e^{\mu/\lambda}-1)\sum_{m=1}^{\infty} m (e^{-\mu/\lambda})^m\\
&= (e^{\mu/\lambda}-1)e^{-\mu/\lambda}/(1-e^{-\mu/\lambda})^2 = 1/(1-e^{-\mu/\lambda}).
\end{align*}
In the following we compute $P(g(k) > y)$. Recall that $I_k = \frac{\lceil \lambda X_{k-1} \rceil}{\lambda} - X_{k-1}$ (Theorem~\ref{thm:GG11}). Using this and Lemma~\ref{lem:g(k)Distr} we obtain
\begin{align*}
&P(g(k) > y) \\
&= \int_{0}^{d}\P\left\{X_{k} + \ceild - x > y + d - x \right\}f_X(x)dx\\
&\quad + \int_{d}^{\infty}\P\left\{X_k + \ceild - x > y \right\}f_X(x)dx \\
&= \int_{0}^{d}\P\left\{X_{k} > y + d  - \ceild \right\}f_X(x)dx\\
&\quad + \int_{d}^{\infty}\P\left\{X_k > y + x  - \ceild  \right\}f_X(x)dx \\ 
&= A + B
\end{align*}
We compute the terms $A$ and $B$ below, and use $\mathbb{E}[g(k)] = \int_{0}^{\infty}P(g(k) > y) dy$ to obtain the result. 

}

\end{document}